\documentclass[12pt,a4paper]{article}

\usepackage{amsmath,amssymb}
\usepackage{mathtools}
\usepackage{graphicx}
\usepackage[nosort]{cite}

\usepackage{tikz}
\usetikzlibrary{decorations.markings}
\usepackage{epsfig}  
\usepackage{enumerate}
\usepackage{relsize}
\usepackage{dsfont}
\usetikzlibrary{calc} 
\usetikzlibrary{patterns} 

\newcommand{\mathsym}[1]{{}}

\usepackage{amsthm,amsbsy,color,multicol}
\usepackage{array,calc}
\usepackage{rotating}
\usepackage{a4wide,bm}
\usepackage{bbm}
\usepackage[a4paper,text={450pt,650pt},centering]{geometry}
\usepackage{setspace}

\let\oldbfseries=\bfseries
\let\oldmdseries=\mdseries
\let\oldnormalfont=\normalfont
\renewcommand{\bfseries}{\oldbfseries\boldmath}
\renewcommand{\mdseries}{\oldmdseries\unboldmath}
\renewcommand{\normalfont}{\oldnormalfont\unboldmath}

\allowdisplaybreaks[3]

\numberwithin{equation}{section}

\usepackage[font=small,labelfont=bf,width=0.85\textwidth]{caption}

\ifx\hypersetup\sadfkjashdfkxja\newcommand\hypersetup[1]{}\fi

\hypersetup{plainpages=false}
\hypersetup{pdfpagemode=UseNone}
\hypersetup{bookmarksnumbered=true}
\hypersetup{pdfstartview=FitH}
\hypersetup{colorlinks=false}
\hypersetup{citebordercolor={.5 1 .5}}
\hypersetup{urlbordercolor={.5 1 1}}
\hypersetup{linkbordercolor={1 .7 .7}}

\DeclareMathSymbol{\Gamma}{\mathalpha}{letters}{"00}
\DeclareMathSymbol{\Delta}{\mathalpha}{letters}{"01}
\DeclareMathSymbol{\Theta}{\mathalpha}{letters}{"02}
\DeclareMathSymbol{\Lambda}{\mathalpha}{letters}{"03}
\DeclareMathSymbol{\Xi}{\mathalpha}{letters}{"04}
\DeclareMathSymbol{\Pi}{\mathalpha}{letters}{"05}
\DeclareMathSymbol{\Sigma}{\mathalpha}{letters}{"06}
\DeclareMathSymbol{\Upsilon}{\mathalpha}{letters}{"07}
\DeclareMathSymbol{\Phi}{\mathalpha}{letters}{"08}
\DeclareMathSymbol{\Psi}{\mathalpha}{letters}{"09}
\DeclareMathSymbol{\Omega}{\mathalpha}{letters}{"0A}

\newcommand{\gen}[1]{\mathrm{#1}}

\newcommand{\dd}{\mathrm{d}}
\newcommand{\ii}{\mathrm{i}}

\newcommand*\widebar[1]{%
  \hbox{%
    \vbox{%
      \hrule height 0.5pt 
      \kern0.25ex
      \hbox{%
        \kern-0.3em
        \ensuremath{#1}%
        \kern-0.1em
      }%
    }%
  }%
}

\ifx\genfrac\sdflkaj\else\fi

\newcommand{\ket}[1]{\left|#1\right\rangle}      
\newcommand{\bra}[1]{\left\langle #1\right|}     
\newcommand{\braket}[2]{\left\langle #1 \right. | \left.  #2 \right\rangle}

\newcommand{\alg}[1]{\mathfrak{#1}}

\newcommand{\beq}{\begin{equation}}
\newcommand{\eeq}{\end{equation}}

\def\[{\begin{equation}}
\def\]{\end{equation}}
\def\<{\begin{eqnarray}}
\def\>{\end{eqnarray}}

\newtheorem{mydef}{Definition}

\newtheorem{theorem}{Theorem}
\newtheorem{lemma}{Lemma} 
\newtheorem{remark}{Remark}
\newtheorem{proposition}{Proposition}

\makeatletter
\def\mr@ignsp#1 {\ifx\:#1\@empty\else #1\expandafter\mr@ignsp\fi}%
\newcommand{\multiref}[1]{\begingroup
\xdef\mr@no@sparg{\expandafter\mr@ignsp#1 \: }%
\def\mr@comma{}%
\@for\mr@refs:=\mr@no@sparg\do{\mr@comma\def\mr@comma{,}\ref{\mr@refs}}%
\endgroup}
\makeatother

\newcommand{\hypref}[2]{\ifx\href\asklfhas #2\else\href{#1}{#2}\fi}
\newcommand{\Secref}[1]{Section~\multiref{#1}}

\newcommand{\Appref}[1]{Appendix~\multiref{#1}}

\newcommand{\Figref}[1]{Figure~\multiref{#1}}

\renewcommand{\eqref}[1]{(\multiref{#1})}

\makeatletter
\newlength{\apb@width}
\newcommand{\autoparbox}[2][c]{\settowidth{\apb@width}{#2}\parbox[#1]{\apb@width}{#2}}

\makeatother

\ifx\href\asklfhas\newcommand{\href}[2]{#2}\fi

\begin{document}

\renewcommand{\thefootnote}{\fnsymbol{footnote}}
\thispagestyle{empty}
\begin{flushright}\footnotesize
ZMP-HH/14-25
\end{flushright}
\vspace{1cm}

\begin{center}%
{\Large\bfseries%
\hypersetup{pdftitle={Off-shell scalar products for the $XXZ$ chain with open boundaries}}%
Off-shell scalar products for \\ the  $XXZ$ spin chain with \\ open boundaries%
\par} \vspace{2cm}%

\textsc{W. Galleas}\vspace{5mm}%
\hypersetup{pdfauthor={Wellington Galleas}}%

\textit{II. Institut f\"ur Theoretische Physik \\ Universit\"at Hamburg, Luruper Chaussee 149 \\ 22761 Hamburg, Germany}\vspace{3mm}%

\verb+wellington.galleas@desy.de+ %

\par\vspace{3cm}

\textbf{Abstract}\vspace{7mm}

\begin{minipage}{12.7cm}
In this work we study scalar products of Bethe vectors associated with the $XXZ$ spin chain with 
open boundary conditions. The scalar products are obtained as solutions of 
a system of functional equations. The description of scalar products through functional relations follows
from a particular map having the \textit{reflection algebra} as its domain and a \textit{function space} as the codomain.
Within this approach we find a multiple contour integral representation for the scalar products in which the homogeneous limit can be 
obtained trivially.

\hypersetup{pdfkeywords={Scalar product, integral formula, open boundaries}}%
\hypersetup{pdfsubject={}}%
\end{minipage}
\vskip 1.5cm
{\small PACS numbers:  05.50+q, 02.30.IK}
\vskip 0.1cm
{\small Keywords: Scalar product, integral formula, open boundaries}
\vskip 2cm
{\small January 2015}

\end{center}

\newpage
\renewcommand{\thefootnote}{\arabic{footnote}}
\setcounter{footnote}{0}

\tableofcontents

\section{Introduction}
\label{sec:intro}

The \textit{inverse scattering method} was formulated in the late sixties by Gardner, Greene, Kruskal and Miura
\cite{inverse_scattering_1967} and it represented a large step towards the understanding of evolution equations. 
This method was originally conceived for solving the Korteweg-deVries (KdV) equation \cite{KdV_1895} but turned out to be
a very general framework, capable of integrating a large number of non-linear differential equations.
Although the application of this method is restricted to a special class of differential equations, its importance
is twofold: on the one hand, it offers a fertile soil where novel and sophisticated mathematical
structures have emerged and continue to do so; on the other hand, the inverse scattering provides a powerful tool for the exact
analysis of physical quantities. For instance, although numerical studies of the KdV equation
were already available in the sixties, the exact solution obtained in \cite{inverse_scattering_1967}
elucidated several aspects concerning the \textit{solitonic} behavior reported in \cite{Zabusky_Kruskal_1965}.
The history repeated itself in \cite{Takhtajan_Faddeev_1981} where the spin of a spin wave
was finally understood through the exact solution of the $XXX$ spin chain obtained by means
of the Bethe ansatz \cite{Bethe_1931}.
The fundamental difference between the non-linear waves governed by the KdV equation and the spin waves
of the $XXX$ chain, is that the first is described by a partial differential equation while the latter
arises from the spectrum of excitations of a quantum mechanical hamiltonian operator. The hamiltonian
of the $XXX$ model is not defined as a differential operator; hence its spectral problem is not a \textit{priori}
characterized by a differential equation. Nevertheless, the spectrum of the $XXX$ spin chain was shown by Baxter to 
admit a description in terms of functional equations \cite{Baxter_1971}.

The existence of a direct correspondence between Baxter's functional equations and integrable non-linear differential
equations is not apparent at first sight but several results suggest this relation should exist. This possibility has been
concluded from different approaches and we refer the reader to  
\cite{Givental_kim_1995, Mukhin_Tarasov_Varchenko_2014, Alexandrov_2013, Zabrodin_2013, Zabrodin_2014, Gorsky_2014, Galleas_2014}
and references therein. In particular, in the work \cite{Galleas_2014} we have presented a first principle mechanism able to convert the spectral problem of spin chains built from solutions of the Yang-Baxter equation 
\cite{Baxter_1971, Sk_Faddeev_1979, Takh_Faddeev_1979} into the solution of linear partial differential equations.
This mechanism is one of the outcomes of the \textit{Algebraic-Functional Method} 
proposed in \cite{Galleas_2008} and refined in the series of papers \cite{Galleas_2010, Galleas_2011, Galleas_2012, Galleas_SCP, Galleas_Twists}.

Among the results obtained through the algebraic-functional approach, we have demonstrated in \cite{Galleas_SCP}
how scalar products of Bethe vectors can be described by functional equations. More precisely, in \cite{Galleas_SCP} we have shown
how the Yang-Baxter algebra can be exploited in order to derive functional relations determining the scalar products associated
with the $XXZ$ chain with periodic boundary conditions. The algebraic structure underlying the $XXZ$ chain with open boundary
conditions is the so called \textit{reflection algebra} \cite{Sklyanin_1988} and here we intend to extend the results of \cite{Galleas_SCP} for 
that case. The reflection algebra is the analogue of the Yang-Baxter algebra for open spin chains and in 
\cite{Galleas_Lamers_2014} we have already demonstrated the feasibility of exploiting it along the lines of \cite{Galleas_2008}.

Scalar products of Bethe vectors are building blocks of correlation functions within the framework of the algebraic 
Bethe ansatz \cite{Korepin_1982, Korepin_book, Kitanine_1999} and, as far as the $XXX$ model with diagonal open boundaries
is concerned, those scalar products have been previously studied in \cite{Wang_2002} using the 
method of \cite{Kitanine_1999}. The results of \cite{Wang_2002} are expressed as determinants and they have
been generalized to the $XXZ$ chain in \cite{Kitanine_2007}. Also, these results have been used in \cite{Wang_2003, Kitanine_2007, Kitanine_2008}
to evaluate certain correlation functions for the $XXX$ and $XXZ$ models with open boundaries. In the works
\cite{Kitanine_2007, Kitanine_2008}, in particular, the authors have obtained multiple integral representations for correlation functions in the case of
a half-infinite lattice. It is also important to remark here that for half-infinite lattices such correlation functions have been
studied in \cite{Kedem_1995a, Kedem_1995b} through the vertex-operator approach and in \cite{Baseilhac_2013, Baseilhac_Kojima_2014a, Baseilhac_Kojima_2014b}
using the $q$-Onsager algebra. 

Here we use the terminology \textit{scalar product of Bethe vectors} but we should keep in mind the conditions in which 
such quantities are computed. Bethe vectors are parameterized by a set of complex variables which are required to satisfy
a set of algebraic equations, i.e. Bethe ansatz equations, in order to having a model eigenvector. The dual Bethe vector required
to compute the aforementioned scalar products also carries another set of complex variables. In this way we have the so called
\textit{off-shell} scalar products when both sets of variables parameterizing the scalar product are arbitrary complex variables.
On the other hand, the cases where only one or both sets of variables satisfy Bethe ansatz equations are usually refereed to as 
\textit{on-shell} scalar products. That being said, we should also remark here that the scalar products obtained in \cite{Wang_2002, Kitanine_2007}
are on-shell. The case of off-shell scalar products for open spin chains has not been considered in the literature so far to the
best of our knowledge, and we shall address this problem in the present paper. Although the computation of correlation functions
usually require on-shell scalar products, the off-shell case can still be regarded as a partition function of a vertex model
with special boundary conditions \cite{Korepin_1982, Tsuchiya_1998} as discussed in \cite{deGier_Galleas_2011}.
More precisely, here we shall present a multiple integral representation for the off-shell scalar product of the 
$XXZ$ spin chain with open boundary conditions. Interestingly, this integral representation allows one to obtain the
homogeneous limit in a trivial way.

The $XXZ$ spin-chain with open boundaries has been extensively discussed in the literature and
this paper has been organized in such a manner to avoid large overlaps with the existing literature.
Therefore, we shall simply collect the required definitions in \Secref{sec:XXZ} in order to clarify
our notation. \Secref{sec:FUN} is then devoted to the derivation of functional relations characterizing
the desired scalar products. The solution of our functional equations is then obtained in \Secref{sec:SOL}. 
Concluding remarks are discussed in \Secref{sec:CONCLUSION} and appendices \ref{sec:SYM} through \ref{sec:Ln1} are
devoted to technical details and proofs.

\section{The open $XXZ$ model and Bethe vectors}
\label{sec:XXZ}

The anisotropic Heisenberg chain was proposed by Bloch in \cite{Bloch_1930, Bloch_1932} as a model for 
remanent magnetization. The exact solution of the $XXZ$ case was then obtained by Yang and Yang through
Bethe ansatz in the series of papers \cite{Yang_Yang_1966, Yang_Yang_1966I, Yang_Yang_1966II, Yang_Yang_1966III}.
In 1967 Sutherland \cite{Sutherland_1967} found a vertex model transfer matrix exhibiting the same eigenvectors firstly obtained 
by Yang and Yang for the $XXZ$ spin chain. It is worth remarking here that, at that time, 
the relation between transfer matrices and spin chains solvable by Bethe ansatz was not clear yet.
The aforementioned works considered the case with periodic boundary conditions while the case
with parallel boundary fields was then solved in \cite{Alcaraz_1987} through a generalization of the Bethe ansatz. 

The model hamiltonian is a linear operator $\mathcal{H} \colon (\mathbb{C}^2)^{\otimes L} \to (\mathbb{C}^2)^{\otimes L}$ with 
$L \in \mathbb{Z}_{>0}$ and it explicitly reads 
\[
\label{ham}
\mathcal{H} \coloneqq \sum_{i=1}^{L-1} \sum_{\alpha \in \{x,y,z \}} J_{\alpha} \; \sigma_{i}^{\alpha} \sigma_{i+1}^{\alpha} + \sinh{(\gamma)} \coth{(h)} \; \sigma_{1}^{z} - \sinh{(\gamma)} \coth{(\bar{h})} \; \sigma_{L}^{z}  
\]
where $J_{x} = J_{y} =1$ and $J_z = \cosh{(\gamma)}$. Here $h$, $\bar{h}$ and $\gamma$ are arbitrary complex parameters while
$\sigma_i^x$, $\sigma_i^y$ and $\sigma_i^z$ are standard Pauli matrices acting on the $i$-th node of the tensor product space
$(\mathbb{C}^2)^{\otimes L}$. The hamiltonian (\ref{ham}) can be embedded as the derivative of a commuting transfer matrix at a particular point.
This construction will be discussed in what follows and here we shall also consider the definitions and notation employed
in \cite{Galleas_Lamers_2014}. Moreover, we shall restrict ourselves to presenting only the required definitions which have not been
described in \cite{Galleas_Lamers_2014}.

\paragraph{Reflection matrices.} Let $\mathcal{K}, \bar{\mathcal{K}} \colon \mathbb{C} \to \gen{End}(\mathbb{C}^2)$ be respectively
refereed to as \textit{reflection matrix} and \textit{dual reflection matrix}. Then following \cite{Sklyanin_1988}, the construction of
integrable spin chains with open boundary conditions through the \textit{Quantum Inverse Scattering Method} (QISM) requires 
$\mathcal{K}$ and $\bar{\mathcal{K}}$ to satisfy the so called \textit{reflection equations}. In this way, $\mathcal{K}$ is governed by the equation
\<
\label{Req}
&& \mathcal{R}_{12}(\lambda_1 - \lambda_2) \mathcal{K}_1 (\lambda_1) \mathcal{R}_{12}(\lambda_1 + \lambda_2) \mathcal{K}_2 (\lambda_2) \nonumber \\
&& = \mathcal{K}_2 (\lambda_2) \mathcal{R}_{12}(\lambda_1 + \lambda_2) \mathcal{K}_1 (\lambda_1) \mathcal{R}_{12}(\lambda_1 - \lambda_2) \; ,
\>
while $\bar{\mathcal{K}}$ is required to satisfy the following dual relation
\<
\label{bReq}
&& \mathcal{R}_{12}(-\lambda_1 + \lambda_2) \bar{\mathcal{K}}_1^{t_1} (\lambda_1) \mathcal{R}_{12}(-\lambda_1 - \lambda_2 - 2\gamma) \bar{\mathcal{K}}_2^{t_2} (\lambda_2) \nonumber \\
&& = \bar{\mathcal{K}}_2^{t_2} (\lambda_2) \mathcal{R}_{12}(-\lambda_1 - \lambda_2 - 2\gamma) \bar{\mathcal{K}}_1^{t_1} (\lambda_1) \mathcal{R}_{12}(-\lambda_1 + \lambda_2) \; .
\>
Both relations (\ref{Req}) and (\ref{bReq}) involve the operator-valued functions $\mathcal{R}_{ij} \colon \mathbb{C} \rightarrow \gen{End}(\mathbb{V}_i \otimes \mathbb{V}_j)$
and $\mathcal{K}_i, \bar{\mathcal{K}}_i \colon \mathbb{C} \to \gen{End}(\mathbb{V}_i \otimes \gen{id})$.
Here we have $\mathbb{V}_i \cong \mathbb{C}^2$ as we are considering the $XXZ$ spin chain hamiltonian (\ref{ham}). Also, the operator $\mathcal{R}_{12}$ corresponds to the $\mathcal{U}_q [\widehat{\alg{sl}}(2)]$ 
invariant $\mathcal{R}$-matrix satisfying the Yang-Baxter equation as described in \cite{Galleas_Lamers_2014}.
In addition to that, the symbol $t_i$ stands for the standard transposition on the space $\mathbb{V}_i$ of a generic matrix
in $\gen{End}( \mathbb{V}_i \otimes \mathbb{V}_j )$.

Here we shall consider the following solution of (\ref{Req}),
\[
\label{kmat}
\mathcal{K}(\lambda) = \left( \begin{matrix}
\sinh{(h + \lambda)} & 0 \cr
0 & \sinh{(h - \lambda)} \end{matrix} \right) \; ,
\]
where $h \in \mathbb{C}$ corresponds to an arbitrary parameter governing the field strength at one of
the boundaries of the hamiltonian (\ref{ham}). The dual reflection matrix $\bar{\mathcal{K}}$ is obtained with the help of the following lemma.

\begin{lemma}[Sklyanin] \label{mapSK}
The map $\mathfrak{d} \left( \mathcal{K}_i (\lambda) \right) \coloneqq \left. \mathcal{K}_i^{t_i} (-\lambda - \gamma) \right|_{h \to \bar{h}}$
is an isomorphism between (\ref{Req}) and (\ref{bReq}).
\end{lemma}
\begin{proof}
We apply the map $\mathfrak{d}$ to (\ref{bReq}) taking into account the following properties satisfied by the $\mathcal{U}_q [\widehat{\alg{sl}}(2)]$ 
$\mathcal{R}$-matrix: $\mathcal{R}_{21} = \mathcal{R}_{12}$, $\mathcal{R}_{12}^{t_1} = \mathcal{R}_{12}^{t_2}$, $\mathcal{R}_{12}(\lambda) \mathcal{R}_{12}(- \lambda) \propto \gen{id}\otimes \gen{id}$
and $\mathcal{R}_{12}^{t_1} (\lambda) \mathcal{R}_{12}^{t_1} (-\lambda - 2\gamma) \propto \gen{id}\otimes \gen{id}$.
\end{proof}

Lemma \ref{mapSK} establishes a map $\mathfrak{d} \colon \mathcal{K} \mapsto \bar{\mathcal{K}}$ which can be exploited to determine
a solution of (\ref{bReq}) in a straightforward manner. From (\ref{kmat}) we then find the dual reflection matrix 
\[
\label{bkmat}
\bar{\mathcal{K}}(\lambda) \coloneqq \mathfrak{d}(\mathcal{K}(\lambda)) = \left( \begin{matrix}
\sinh{(\bar{h} - \lambda - \gamma)} & 0 \cr
0 & \sinh{(\bar{h} + \lambda + \gamma)} \end{matrix} \right) \; .
\]

\paragraph{Double-row transfer matrix.} Let $\lambda, \mu_j \in \mathbb{C}$ be respectively refereed to as spectral parameter and inhomogeneity
parameter. Then the \textit{double-row transfer matrix} or simply \textit{transfer matrix} $T \colon \mathbb{C} \to \gen{End}( (\mathbb{C}^2)^{\otimes L} )$ 
is defined as the following operator,
\[
\label{tmat}
T(\lambda) \coloneqq \gen{tr}_{0} \left[ \bar{\mathcal{K}}_0 (\lambda)  \mathop{\overleftarrow\prod}\limits_{1 \leq j \leq L} \mathcal{R}_{0 j}(\lambda - \mu_j) \; \mathcal{K}_0 (\lambda) \mathop{\overrightarrow\prod}\limits_{1 \leq j \leq L} \mathcal{R}_{0 j}(\lambda + \mu_j) \right] \; .
\]
The hamiltonian (\ref{ham}) then corresponds to
\[
\mathcal{H} = \left[ 2 \sinh{(\gamma)}^{2L} \coth{(\gamma)} \sinh{(h)} \sinh{(\bar{h})} \right]^{-1} \left. \frac{d T(\lambda)}{d \lambda} \right|_{\stackrel{\lambda = 0}{\mu_j = 0}} - L \cosh{(\gamma)} - \sinh{(\gamma)} \tanh{(\gamma)} \; .
\]

\paragraph{ABCD structure.} The trace $\gen{tr}_{0}$ in (\ref{tmat}) is taken only over the space $\mathbb{V}_0$, while
the term inside the brackets lives in $\gen{End} (\mathbb{V}_0 \otimes \mathbb{V}_{\mathcal{Q}})$ where 
$\mathbb{V}_{\mathcal{Q}} \cong (\mathbb{C}^2))^{\otimes L}$. Thus, it is convenient to employ the following representation, 
\[
\label{ABCD}
\mathop{\overleftarrow\prod}\limits_{1 \leq j \leq L} \mathcal{R}_{0 j}(\lambda - \mu_j) \; \mathcal{K}_0 (\lambda) \mathop{\overrightarrow\prod}\limits_{1 \leq j \leq L} \mathcal{R}_{0 j}(\lambda + \mu_j) \eqqcolon
\left(  \begin{matrix}
\mathcal{A}(\lambda) & \mathcal{B}(\lambda) \cr
\mathcal{C}(\lambda) & \mathcal{D}(\lambda) \end{matrix} \right) \; ,
\]
where $\mathcal{A}, \mathcal{B} , \mathcal{C} , \mathcal{D} \in \gen{End}(\mathbb{V}_{\mathcal{Q}})$. In terms of these operators,  
the transfer matrix (\ref{tmat}) reads
\[
\label{tmatAD}
T(\lambda) = \sinh{(\bar{h} - \lambda - \gamma)} \mathcal{A}(\lambda) + \sinh{(\bar{h} + \lambda + \gamma)} \mathcal{D}(\lambda) \; .
\]

\paragraph{Highest-weight property.} The highest-weight representation theory of the $\alg{sl}(2)$ Lie algebra plays a prominent role 
for the system we are considering. For instance, the vector
\[
\label{zero}
\ket{0} \coloneqq \left( \begin{matrix} 1 \cr 0 \end{matrix} \right)^{\otimes L}
\]
is a $\alg{sl}(2)$ highest-weight vector in $\mathbb{V}_{\mathcal{Q}}$. Its dual is simply
given by $\bra{0} \coloneqq (1 \quad 0)^{\otimes L}$.  The left and right action of the operators (\ref{ABCD})
on these vectors can be straightforwardly computed due to the structure of the $\mathcal{U}_q [\widehat{\alg{sl}}(2)]$ 
invariant $\mathcal{R}$-matrix entering in the definition (\ref{ABCD}). In this way we have the following expressions:
\begin{align}
\label{action}
\mathcal{A}(\lambda) \ket{0} =& \Lambda_{\mathcal{A}} (\lambda) \ket{0} & \tilde{\mathcal{D}}(\lambda) \ket{0} & = \Lambda_{\tilde{\mathcal{D}}} (\lambda) \ket{0} \cr
\bra{0} \mathcal{A}(\lambda) =& \Lambda_{\mathcal{A}} (\lambda) \bra{0} & \bra{0} \tilde{\mathcal{D}}(\lambda) & = \Lambda_{\tilde{\mathcal{D}}} (\lambda) \bra{0} \cr
\mathcal{C}(\lambda) \ket{0} =& 0 &  \bra{0} \mathcal{B}(\lambda) & = 0 \; , 
\end{align}
where $\tilde{\mathcal{D}}(\lambda) \coloneqq \mathcal{D}(\lambda) - \frac{\sinh{(\gamma)}}{\sinh{(2 \lambda + \gamma)}} \mathcal{A}(\lambda)$ 
has been defined for latter convenience. It is also useful to define the functions $a(\lambda) \coloneqq \sinh{(\lambda + \gamma)}$, 
$b(\lambda) \coloneqq \sinh{(\lambda)}$ and $c(\lambda) \coloneqq \sinh{(\gamma)}$, in such a way that the coefficients $\Lambda_{\mathcal{A}}$ and
$\Lambda_{\tilde{\mathcal{D}}}$  explicitly read
\begin{eqnarray}
\label{lambda}
\Lambda_{\mathcal{A}} (\lambda) & \coloneqq & b(h + \lambda) \prod_{j=1}^{L} a(\lambda - \mu_j) a(\lambda + \mu_j)  \nonumber \\
\Lambda_{\tilde{\mathcal{D}}} (\lambda) &\coloneqq & - \frac{b(2 \lambda)}{a(2 \lambda)} a(\lambda - h) \prod_{j=1}^{L} b(\lambda - \mu_j) b(\lambda + \mu_j)   \; .
\end{eqnarray}

\paragraph{Bethe vectors.} One of the achievements of the QISM is the algebraic construction of eigenvectors of commuting
transfer matrices. This is also the case for spin chains with open boundary conditions as shown in \cite{Sklyanin_1988}.
In particular, the eigenvectors associated with the spin chain hamiltonian (\ref{ham}) can be built with the help of the operators
defined in (\ref{ABCD}).
Those vectors $\ket{\gen{\Psi}_n} \in \gen{span}(\mathbb{V}_{\mathcal{Q}})$ will be refereed to as \textit{off-shell Bethe vectors} and
they are defined as
\[
\label{BV}
\ket{\gen{\Psi}_n} \coloneqq \mathop{\overrightarrow\prod}\limits_{1 \leq j \leq n} \mathcal{B}(\lambda_j^{\mathcal{B}}) \ket{0} \; .
\]
Strictly speaking, the vector $\ket{\gen{\Psi}_n}$ is only an eigenvector of the transfer matrix (\ref{tmatAD}) for particular
choices of the set of variables $\{ \lambda_j^{\mathcal{B}} \}$. Dual eigenvectors can also be built in a similar way and we shall
refer to them as \textit{Dual Bethe vectors}. They are defined as
 \[
\label{dBV}
\bra{\gen{\Psi}_n} \coloneqq \bra{0} \mathop{\overleftarrow\prod}\limits_{1 \leq j \leq n} \mathcal{C}(\lambda_j^{\mathcal{C}}) \; .
\]
Here we shall assume that the set of variables $\{ \lambda_j^{\mathcal{C}} \}$ is independent of the set $\{ \lambda_j^{\mathcal{B}} \}$ in order
to keep our results as general as possible.

\paragraph{Scalar product.} The quantity we are interested in the present work is the scalar product of Bethe vectors
$\mathcal{S}_n \colon \mathbb{C}^{2n} \to \mathbb{C}$ defined as,
\<
\label{scp}
\mathcal{S}_n (\lambda_1^{\mathcal{C}}, \dots , \lambda_n^{\mathcal{C}} | \lambda_1^{\mathcal{B}}, \dots , \lambda_n^{\mathcal{B}} ) &\coloneqq& \braket{\gen{\Psi}_n}{\gen{\Psi}_n} \nonumber \\
&=& \bra{0} \mathop{\overleftarrow\prod}\limits_{1 \leq i \leq n} \mathcal{C}(\lambda_i^{\mathcal{C}}) \mathop{\overrightarrow\prod}\limits_{1 \leq j \leq n} \mathcal{B}(\lambda_j^{\mathcal{B}})  \ket{0} \; .
\>
The quantity $\mathcal{S}_n$ is usually refereed to as off-shell scalar product as the variables $\lambda_j^{\mathcal{B}, \mathcal{C}}$
are arbitrary. The on-shell case corresponds to the situation where the variables $\lambda_j^{\mathcal{B}}$, and/or $\lambda_j^{\mathcal{C}}$,
are constrained by Bethe ansatz equations. Here we are interested in the off-shell case, but for completeness
we shall also present the on-shell constraint. It reads
\<
\label{BA}
&&\theta(\lambda_j^{\mathcal{B}}, h ) \theta(\lambda_j^{\mathcal{B}}, -\bar{h} ) \prod_{k=1}^L \frac{\sinh{(\lambda_j^{\mathcal{B}} - \mu_k + \gamma)}}{\sinh{(\lambda_j^{\mathcal{B}} - \mu_k)}} \frac{\sinh{(\lambda_j^{\mathcal{B}} + \mu_k + \gamma)}}{\sinh{(\lambda_j^{\mathcal{B}} + \mu_k)}} \nonumber \\
&& \qquad \qquad \qquad \quad =  \prod_{\stackrel{l=1}{l \neq j}}^n \frac{\sinh{(\lambda_j^{\mathcal{B}} - \lambda_l^{\mathcal{B}} + \gamma)}}{\sinh{(\lambda_j^{\mathcal{B}} - \lambda_l^{\mathcal{B}} - \gamma)}} \frac{\sinh{(\lambda_j^{\mathcal{B}} + \lambda_l^{\mathcal{B}} + \gamma)}}{\sinh{(\lambda_j^{\mathcal{B}} + \lambda_l^{\mathcal{B}} - \gamma)}} \; ,
\>
where $\theta (\lambda, \omega) \coloneqq \frac{\sinh{(\lambda + \omega)}}{\sinh{(\lambda - \omega + \gamma)}}$. 
The solutions of (\ref{BA}) are usually refereed to as \textit{Bethe roots}.
Also, dual Bethe vectors (\ref{dBV}) will only be eigenvectors of (\ref{tmatAD}) when the variables $\lambda_j^{\mathcal{C}}$ are 
Bethe roots. The scalar product $\mathcal{S}_n$ is a key ingredient for the computation of correlation functions \cite{Wang_2002, Kitanine_2007, Kitanine_2008}
and in what follows we shall investigate this quantity along the lines of \cite{Galleas_SCP}.

\section{Reflection algebra and functional equations}
\label{sec:FUN}

The operators $\mathcal{A}$, $\mathcal{B}$, $\mathcal{C}$ and $\mathcal{D}$ defined
in (\ref{ABCD}) are subjected to the reflection algebra relations \cite{Sklyanin_1988}.
See also \cite{Galleas_Lamers_2014} for the conventions we are using here. This property is
a direct consequence of the Yang-Baxter equation \cite{Baxter_book}, in addition to the 
reflection equation (\ref{Req}). In the present paper we aim to use the reflection 
algebra relations as a source of functional equations characterizing the scalar product 
(\ref{scp}). For that it is convenient to introduce the following definitions.

\begin{mydef}[Higher-degree monodromy set]
Let $\mathcal{M}(\lambda) \coloneqq \{ \mathcal{A}, \mathcal{B} , \mathcal{C} , \mathcal{D} \}(\lambda)$ for $\lambda \in \mathbb{C}$
be a set whose elements are the operators defined in (\ref{ABCD}). Then define the monodromy set of degree $n$ as $\mathcal{W}_n \coloneqq \mathcal{M}(\lambda_1) \otimes \mathcal{M}(\lambda_2) \otimes \dots \mathcal{M}(\lambda_n)$
and $\widetilde{\mathcal{W}}_n \coloneqq \mathbb{C}[\lambda_1^{\pm 1} , \lambda_2^{\pm 1} , \dots , \lambda_n^{\pm 1}] \otimes \mathcal{W}_n$.
\end{mydef}  

\begin{mydef}[Higher-degree reflection relation]
Let $\mathfrak{SK}_2 \subset \widetilde{\mathcal{W}}_2$ be the set of quadratic relations originated from the reflection algebra as described in
\cite{Galleas_Lamers_2014}. We then define the \textit{reflection relations of degree $n$} as the elements of $\mathfrak{SK}_n \subset \widetilde{\mathcal{W}}_n$
defined recursively through the relation
\[
\mathfrak{SK}_n \simeq \frac{\mathfrak{SK}_{n-1} \otimes \mathcal{M}(\lambda_n)}{\mathfrak{SK}_2} \; .
\]
In other words, the relations in $\mathfrak{SK}_n$ consist of the expressions obtained from $\mathcal{W}_n$ $(n>2)$ after the repeated use
of the relations in $\mathfrak{SK}_2$ \footnote{The author thanks J. Lamers for useful discussions on this point.}.
\end{mydef}

The set $\mathfrak{SK}_2$ is formed by the fundamental reflection algebra relations and they amount to 
sixteen relations in the case considered here. We shall only need a few relations in $\mathfrak{SK}_2$
for our purposes. For instance, we have $\mathcal{B}(\lambda_1) \mathcal{B}(\lambda_2) = \mathcal{B}(\lambda_2) \mathcal{B}(\lambda_1)$ and 
$\mathcal{C}(\lambda_1) \mathcal{C}(\lambda_2) = \mathcal{C}(\lambda_2) \mathcal{C}(\lambda_1)$ among them. These two relations
motivates the use of the following simplified notation.

\begin{mydef}
Let $Z \coloneqq \{ \lambda_i \in \mathbb{C} \; | \; 1 \leq i \leq n \}$ be a set of cardinality $n$. 
We then define the symbols $[ Z ]_{\mathcal{B}}$ and $[ Z ]_{\mathcal{C}}$ as
\[
\label{box}
[ Z ]_{\mathcal{B}} \coloneqq \mathop{\overrightarrow\prod}\limits_{1 \leq i \leq n} \mathcal{B}(\lambda_i) \quad \mbox{and} \quad [ Z ]_{\mathcal{C}} \coloneqq \mathop{\overleftarrow\prod}\limits_{1 \leq i \leq n} \mathcal{C}(\lambda_i) \; .
\]
We can safely use the notation (\ref{box}) since the operators $\mathcal{B}$'s commute for different values of their spectral parameters.
The same argument also applies for the operators $\mathcal{C}$'s. 
\end{mydef}

As far as the evaluation of the scalar product $\mathcal{S}_n$ is concerned, we shall make use of the following relations
in $\mathfrak{SK}_{n+1}$:
\<
\label{AB}
\mathcal{A}(\lambda_0) [\gen{Y}^{1,n}]_{\mathcal{B}} &=& \prod_{\lambda \in \gen{Y}^{1,n}} \frac{a(\lambda - \lambda_0)}{b(\lambda - \lambda_0)} \frac{b(\lambda + \lambda_0)}{a(\lambda + \lambda_0)} [\gen{Y}^{1,n}]_{\mathcal{B}} \; \mathcal{A}(\lambda_0) \nonumber \\
&-& \sum_{\lambda \in \gen{Y}^{1,n}} [\gen{Y}_{\lambda}^{0,n}]_{\mathcal{B}} \left\{ \frac{c(\lambda - \lambda_0)}{b(\lambda - \lambda_0)} \frac{b(2 \lambda)}{a(2 \lambda)} \prod_{\tilde{\lambda} \in \gen{Y}_{\lambda}^{1,n}} \frac{a(\tilde{\lambda} - \lambda)}{b(\tilde{\lambda} - \lambda)} \frac{b(\tilde{\lambda} + \lambda)}{a(\tilde{\lambda} + \lambda)} \mathcal{A}(\lambda) \right. \nonumber \\
&& \qquad \qquad \qquad \quad \left. + \; \frac{c(\lambda + \lambda_0)}{a(\lambda + \lambda_0)} \prod_{\tilde{\lambda} \in \gen{Y}_{\lambda}^{1,n}} \frac{a(\lambda - \tilde{\lambda})}{b(\lambda - \tilde{\lambda})} \frac{a(\lambda + \tilde{\lambda} + \gamma)}{b(\lambda + \tilde{\lambda} + \gamma)} \tilde{\mathcal{D}}(\lambda) \right\} \nonumber \\
\>
\<
\label{CA}
[\gen{X}^{1,n}]_{\mathcal{C}} \; \mathcal{A}(\lambda_0) &=& \prod_{\lambda \in \gen{X}^{1,n}} \frac{a(\lambda - \lambda_0)}{b(\lambda - \lambda_0)} \frac{b(\lambda + \lambda_0)}{a(\lambda + \lambda_0)} \mathcal{A}(\lambda_0) [\gen{X}^{1,n}]_{\mathcal{C}}  \nonumber \\
&-& \sum_{\lambda \in \gen{X}^{1,n}} \left\{ \frac{c(\lambda - \lambda_0)}{b(\lambda - \lambda_0)} \frac{b(2 \lambda)}{a(2 \lambda)} \prod_{\tilde{\lambda} \in \gen{X}_{\lambda}^{1,n}} \frac{a(\tilde{\lambda} - \lambda)}{b(\tilde{\lambda} - \lambda)} \frac{b(\tilde{\lambda} + \lambda)}{a(\tilde{\lambda} + \lambda)} \mathcal{A}(\lambda) \right. \nonumber \\
&& \qquad \qquad \left. + \; \frac{c(\lambda + \lambda_0)}{a(\lambda + \lambda_0)} \prod_{\tilde{\lambda} \in \gen{X}_{\lambda}^{1,n}} \frac{a(\lambda - \tilde{\lambda})}{b(\lambda - \tilde{\lambda})} \frac{a(\lambda + \tilde{\lambda} + \gamma)}{b(\lambda + \tilde{\lambda} + \gamma)} \tilde{\mathcal{D}}(\lambda) \right\} [\gen{X}_{\lambda}^{0,n}]_{\mathcal{C}} \; . \nonumber \\
\>
In (\ref{AB}) and (\ref{CA}) we have considered the sets $\gen{X}^{i,j}$ and $\gen{Y}^{i,j}$ respectively defined as
$\gen{X}^{i,j} \coloneqq \{ \lambda_k^C \; | \; i \leq k \leq j \}$ and $\gen{Y}^{i,j} \coloneqq \{ \lambda_k^B \; | \; i \leq k \leq j \}$.
In addition to that we have also employed the sets $\gen{X}_{\lambda}^{i,j} \coloneqq \gen{X}^{i,j} \backslash \{ \lambda \}$ and $\gen{Y}_{\lambda}^{i,j} \coloneqq \gen{Y}^{i,j} \backslash \{ \lambda \}$.

The relations (\ref{AB}) and (\ref{CA}) do not exhaust all possible relations in $\mathfrak{SK}_{n+1}$ which can be 
used to determine the scalar product $\mathcal{S}_n$. For instance, here we shall also make use of the following
ones:
\<
\label{DB}
\mathcal{D}(\lambda_0) [\gen{Y}^{1,n}]_{\mathcal{B}} &=& \prod_{\lambda \in \gen{Y}^{1,n}} \frac{a(\lambda_0 - \lambda)}{b(\lambda_0 - \lambda)} 
\frac{a(\lambda_0 + \lambda + \gamma)}{b(\lambda_0 + \lambda + \gamma)} [\gen{Y}^{1,n}]_{\mathcal{B}} \; \mathcal{D}(\lambda_0) \nonumber \\
&+& \sum_{\lambda \in \gen{Y}^{1,n}} [\gen{Y}_{\lambda}^{0,n}]_{\mathcal{B}} \left\{ \frac{c(2 \lambda_0)}{a(2\lambda_0)} \frac{b(2 \lambda)}{a(2 \lambda)} \frac{a(2 \lambda_0 + \gamma)}{a( \lambda_0 + \lambda)} \prod_{\tilde{\lambda} \in \gen{Y}_{\lambda}^{1,n}} \frac{a(\tilde{\lambda} - \lambda)}{b(\tilde{\lambda} - \lambda)} \frac{b(\tilde{\lambda} + \lambda)}{a(\tilde{\lambda} + \lambda)} \mathcal{A}(\lambda) \right. \nonumber \\
&& \qquad \qquad   \left. - \; \frac{a(2 \lambda_0 + \gamma)}{b(2 \lambda_0 + \gamma)} \frac{c(\lambda_0 - \lambda)}{b(\lambda_0 - \lambda)}
\prod_{\tilde{\lambda} \in \gen{Y}_{\lambda}^{1,n}} \frac{a(\lambda - \tilde{\lambda})}{b(\lambda - \tilde{\lambda})} \frac{a(\lambda + \tilde{\lambda} + \gamma)}{b(\lambda + \tilde{\lambda} + \gamma)} \tilde{\mathcal{D}}(\lambda) \right\} \nonumber \\
\>

\<
\label{CD}
[\gen{X}^{1,n}]_{\mathcal{C}} \mathcal{D}(\lambda_0)  &=& \prod_{\lambda \in \gen{X}^{1,n}} \frac{a(\lambda_0 - \lambda)}{b(\lambda_0 - \lambda)} 
\frac{a(\lambda_0 + \lambda + \gamma)}{b(\lambda_0 + \lambda + \gamma)}  \; \mathcal{D}(\lambda_0) [\gen{X}^{1,n}]_{\mathcal{C}} \nonumber \\
&+& \sum_{\lambda \in \gen{X}^{1,n}}  \left\{ \frac{c(2 \lambda_0)}{a(2\lambda_0)} \frac{b(2 \lambda)}{a(2 \lambda)} \frac{a(2 \lambda_0 + \gamma)}{a( \lambda_0 + \lambda)} \prod_{\tilde{\lambda} \in \gen{X}_{\lambda}^{1,n}} \frac{a(\tilde{\lambda} - \lambda)}{b(\tilde{\lambda} - \lambda)} \frac{b(\tilde{\lambda} + \lambda)}{a(\tilde{\lambda} + \lambda)} \mathcal{A}(\lambda) \right. \nonumber \\
&&   \left. - \; \frac{a(2 \lambda_0 + \gamma)}{b(2 \lambda_0 + \gamma)} \frac{c(\lambda_0 - \lambda)}{b(\lambda_0 - \lambda)}
\prod_{\tilde{\lambda} \in \gen{X}_{\lambda}^{1,n}} \frac{a(\lambda - \tilde{\lambda})}{b(\lambda - \tilde{\lambda})} \frac{a(\lambda + \tilde{\lambda} + \gamma)}{b(\lambda + \tilde{\lambda} + \gamma)} \tilde{\mathcal{D}}(\lambda) \right\} [\gen{X}_{\lambda}^{0,n}]_{\mathcal{C}} \; . \nonumber \\
\>

\subsection{Algebraic-functional approach}
\label{sec:FUN}

In this section we aim to derive functional relations determining the scalar product (\ref{scp}). 
The method we shall employ consists of an extension of the work \cite{Galleas_SCP} for the case of
open boundary conditions. The key idea is to use the reflection algebra as a source of functional relations
as discussed in \cite{Galleas_Lamers_2014}. This method was first proposed in \cite{Galleas_2008} for spectral 
problems and extended for vertex models partition functions in the works \cite{Galleas_2010, Galleas_2011, Galleas_2012, Galleas_2013, Galleas_proc}.
We refer to this approach as \textit{Algebraic-Functional Method} and an important step within it is to find
a suitable linear map $\gen{\pi}_n \colon \mathfrak{SK}_n \to \mathbb{C}[\lambda_1^{\pm 1} , \lambda_2^{\pm 1} , \dots , \lambda_n^{\pm 1}]$
able to convert a higher-degree reflection algebra relation into a complex multivariate function. 
Moreover, we would like to build a map $\gen{\pi}_n$ yielding the simplest possible functional equation for the scalar product (\ref{scp}).
In fact, the structure of the resulting functional equation will depend on two ingredients: the higher-degree relation
in $\mathfrak{SK}_n$ we single out and the particular realization of $\gen{\pi}_n$ we are considering.     

\begin{proposition}[The map $\gen{\pi}_n$]
The following scalar products are realizations of $\gen{\pi}_n$,
\<
\label{piBC}
\gen{\pi}_n^{\mathcal{B}} ( \mathfrak{h} ) &\coloneqq&  \bra{0} \mathfrak{h} \; [ \gen{Y}^{1,n} ]_{\mathcal{B}} \ket{0} \qquad \forall \; \mathfrak{h} \in \mathfrak{SK}_n \; , \nonumber \\
\gen{\pi}_n^{\mathcal{C}} ( \mathfrak{h} ) &\coloneqq&  \bra{0} [ \gen{X}^{1,n} ]_{\mathcal{C}} \; \mathfrak{h} \ket{0} \qquad \forall \; \mathfrak{h} \in \mathfrak{SK}_n \; .
\>
\end{proposition}
\begin{proof}
From (\ref{piBC}) we can readily see that $\gen{\pi}_n^{\mathcal{B},\mathcal{C}} \colon \widetilde{\mathcal{W}}_n \to \mathbb{C}[\lambda_1^{\pm 1} , \dots , \lambda_n^{\pm 1}]$.
Therefore, as $\mathfrak{SK}_n \subset \widetilde{\mathcal{W}}_n$, our claim is automatically fulfilled.
\end{proof}

In what follows we shall demonstrate how one can use (\ref{piBC}) to convert the relations (\ref{AB})-(\ref{CD})
into two functional equations characterizing the scalar product $\mathcal{S}_n$.

\subsubsection{Equation type A}
\label{sec:TYPEA}

The functional equation obtained from the reflection algebra relations (\ref{AB}) and (\ref{CA}) 
shall be refereed to as equation type A. This equation is a consequence of the trivial identity
\[
\label{ABCA}
\gen{\pi}_{n+1}^{\mathcal{C}} \left( \mathcal{A}(\lambda_0) [\gen{Y}^{1,n}]_{\mathcal{B}} \right) =  \gen{\pi}_{n+1}^{\mathcal{B}} \left( [\gen{X}^{1,n}]_{\mathcal{C}} \; \mathcal{A}(\lambda_0) \right) \; ,
\]
which is the starting point to prove Theorem \ref{funA}. For convenience we also introduce the following notation.

\begin{mydef} \label{vec}
Let $i,j \in \mathbb{Z}$ such that $i < j$. Then define $\vec{X}^{i,j}$ and $\vec{Y}^{i,j}$ as the following $(j-i+1)$-dimensional vectors,
\<
\label{vecXY}
\vec{X}^{i,j} &\coloneqq& \left( \lambda_i^{\mathcal{C}}, \lambda_{i+1}^{\mathcal{C}}, \dots , \lambda_{j}^{\mathcal{C}}  \right) \nonumber \\
\vec{Y}^{i,j} &\coloneqq& \left( \lambda_i^{\mathcal{B}}, \lambda_{i+1}^{\mathcal{B}}, \dots , \lambda_{j}^{\mathcal{B}}  \right) \; .
\>
Also, consider $\lambda \in \gen{X}^{i,j}$ and $\bar{\lambda} \in \gen{Y}^{i,j}$, and additionally define the $(j-i)$-dimensional vectors
\<
\label{vecXYom}
\vec{X}_{\lambda}^{i,j} &\coloneqq& \left( \lambda_i^{\mathcal{C}}, \lambda_{i+1}^{\mathcal{C}}, \dots , \underline{\lambda} , \dots , \lambda_{j}^{\mathcal{C}}  \right) \nonumber \\
\vec{Y}_{\bar{\lambda}}^{i,j} &\coloneqq& \left( \lambda_i^{\mathcal{B}}, \lambda_{i+1}^{\mathcal{B}}, \dots , \underline{\bar{\lambda}}, \dots , \lambda_{j}^{\mathcal{B}}  \right) \; .
\>
In (\ref{vecXYom}) the underline stands for omission.
\end{mydef}

\begin{theorem} \label{funA}
The scalar product $\mathcal{S}_n$, as defined in (\ref{scp}), satisfies the equation
\[
\label{typeA}
M_0 \; \mathcal{S}_n (\vec{X}^{1,n} | \vec{Y}^{1,n}) + \sum_{\lambda \in \gen{X}^{1,n}} N_{\lambda}^{(\mathcal{C})} \; \mathcal{S}_n (\vec{X}_{\lambda}^{0,n} | \vec{Y}^{1,n}) 
+ \sum_{\lambda \in \gen{Y}^{1,n}} N_{\lambda}^{(\mathcal{B})} \; \mathcal{S}_n (\vec{X}^{1,n} | \vec{Y}_{\lambda}^{0,n}) = 0 \; ,
\]
with coefficients
\<
\label{coeffA}
M_0 &\coloneqq& \Lambda_{\mathcal{A}} (\lambda_0) \left\{ \prod_{\lambda \in \gen{Y}^{1,n}} \frac{a(\lambda - \lambda_0)}{b(\lambda - \lambda_0)} \frac{b(\lambda + \lambda_0)}{a(\lambda + \lambda_0)} - \prod_{\lambda \in \gen{X}^{1,n}} \frac{a(\lambda - \lambda_0)}{b(\lambda - \lambda_0)} \frac{b(\lambda + \lambda_0)}{a(\lambda + \lambda_0)} \right\} \nonumber \\
N_{\lambda}^{(\mathcal{C})} &\coloneqq& \Lambda_{\mathcal{A}} (\lambda) \frac{c(\lambda - \lambda_0)}{b(\lambda - \lambda_0)} \frac{b(2 \lambda)}{a(2 \lambda)} \prod_{\tilde{\lambda} \in \gen{X}_{\lambda}^{1,n}} \frac{a(\tilde{\lambda} - \lambda)}{b(\tilde{\lambda} - \lambda)} \frac{b(\tilde{\lambda} + \lambda)}{a(\tilde{\lambda} + \lambda)} \nonumber \\
&& \qquad + \; \Lambda_{\tilde{\mathcal{D}}}  (\lambda)  \frac{c(\lambda + \lambda_0)}{a(\lambda + \lambda_0)} \prod_{\tilde{\lambda} \in \gen{X}_{\lambda}^{1,n}} \frac{a(\lambda - \tilde{\lambda})}{b(\lambda - \tilde{\lambda})} \frac{a(\lambda + \tilde{\lambda} + \gamma)}{b(\lambda + \tilde{\lambda} + \gamma)} \nonumber \\
N_{\lambda}^{(\mathcal{B})} &\coloneqq& - \Lambda_{\mathcal{A}} (\lambda) \frac{c(\lambda - \lambda_0)}{b(\lambda - \lambda_0)} \frac{b(2 \lambda)}{a(2 \lambda)} \prod_{\tilde{\lambda} \in \gen{Y}_{\lambda}^{1,n}} \frac{a(\tilde{\lambda} - \lambda)}{b(\tilde{\lambda} - \lambda)} \frac{b(\tilde{\lambda} + \lambda)}{a(\tilde{\lambda} + \lambda)} \nonumber \\
&& \qquad - \; \Lambda_{\tilde{\mathcal{D}}} (\lambda)  \frac{c(\lambda + \lambda_0)}{a(\lambda + \lambda_0)} \prod_{\tilde{\lambda} \in \gen{Y}_{\lambda}^{1,n}} \frac{a(\lambda - \tilde{\lambda})}{b(\lambda - \tilde{\lambda})} \frac{a(\lambda + \tilde{\lambda} + \gamma)}{b(\lambda + \tilde{\lambda} + \gamma)} \; .
\>
\end{theorem}
\begin{proof}
Within the algebraic-functional approach, we firstly consider the action of the map $\gen{\pi}_{n+1}^{\mathcal{C}}$, defined
in (\ref{piBC}), on the relation (\ref{AB}). By doing so we are left with terms of the form: 
$\gen{\pi}_{n+1}^{\mathcal{C}} \left( [\gen{Y}^{1,n}]_{\mathcal{B}} \mathcal{A}(\lambda_0) \right)$, $\gen{\pi}_{n+1}^{\mathcal{C}} \left( [\gen{Y}_{\lambda}^{0,n}]_{\mathcal{B}} \mathcal{A}(\lambda) \right)$
and $\gen{\pi}_{n+1}^{\mathcal{C}} \left( [\gen{Y}_{\lambda}^{0,n}]_{\mathcal{B}} \tilde{\mathcal{D}}(\lambda) \right)$. 
Those terms satisfy the reduction property $\gen{\pi}_{n+1}^{\mathcal{C}} \to \gen{\pi}_{n}^{\mathcal{C}}$ due to the highest-weight property of the 
vector $\ket{0}$ entering in the definition (\ref{piBC}).
Next we consider the action of $\gen{\pi}_{n+1}^{\mathcal{B}}$ on the higher-degree relation (\ref{CA}). This procedure 
yields terms of the form $\gen{\pi}_{n+1}^{\mathcal{B}} \left( \mathcal{A}(\lambda_0) \; [\gen{X}^{1,n}]_{\mathcal{C}} \right)$,   
$\gen{\pi}_{n+1}^{\mathcal{B}} \left( \mathcal{A}(\lambda) \; [\gen{X}_{\lambda}^{0,n}]_{\mathcal{C}} \right)$ and $\gen{\pi}_{n+1}^{\mathcal{B}} \left( \tilde{\mathcal{D}}(\lambda) \; [\gen{X}_{\lambda}^{0,n}]_{\mathcal{C}} \right)$.
Similarly, those terms also enjoy the property $\gen{\pi}_{n+1}^{\mathcal{B}} \to \gen{\pi}_{n}^{\mathcal{B}}$ due to the highest-weight property
of (\ref{zero}). The precise relations associated with $\gen{\pi}_{n+1}^{\mathcal{B}, \mathcal{C}} \to \gen{\pi}_{n}^{\mathcal{B}, \mathcal{C}}$
can be found with the help of (\ref{action}). Lastly, we use the identifications 
\[
\label{piS}
\gen{\pi}_{n}^{\mathcal{B}} \left( [\gen{X}^{1,n}]_{\mathcal{C}} \right) = \gen{\pi}_{n}^{\mathcal{C}} \left( [\gen{Y}^{1,n}]_{\mathcal{B}} \right) = \mathcal{S}_n (\vec{X}^{1,n} | \vec{Y}^{1,n}) \; .
\]
The relations (\ref{AB}), (\ref{CA}), (\ref{piBC}) and (\ref{ABCA}) then imply the functional relation (\ref{typeA}).
\end{proof}

Interestingly, Eq. (\ref{typeA}) exhibits the same structure as the equations derived in \cite{Galleas_SCP}
for the scalar products of Bethe vectors associated with the $XXZ$ spin chain with periodic boundary conditions. 
The only difference is in the particular form of the coefficients $M_0$, $N_{\lambda}^{(\mathcal{B})}$ and $N_{\lambda}^{(\mathcal{C})}$.
In \cite{Galleas_SCP} we have presented two functional equations describing the desired scalar product and the situation is similar
in the present paper. A second functional equation satisfied by the scalar product (\ref{scp}) will be described in the next section.

\subsubsection{Equation type D}
\label{sec:TYPED}

As anticipated in the previous section, there exists a second functional equation satisfied by the scalar 
product $\mathcal{S}_n$. This second equation exhibits the same structure as (\ref{typeA}).
In fact, a large number of functional equations could be derived by exploring different reflection algebra relations and different
realizations of the maps $\gen{\pi}_{n}^{\mathcal{B},\mathcal{C}}$. This feature of the algebraic-functional method has been previously
observed for partition functions of vertex models with domain-wall boundaries as discussed in \cite{Galleas_proc}.
However, we would like to deal with the simplest possible equations and a systematic method for solving equations
with the structure of (\ref{typeA}) has been already developed in \cite{Galleas_2012, Galleas_2013, Galleas_SCP}.

This second equation for the scalar product $\mathcal{S}_n$ will be obtained from the higher-degree relations
(\ref{DB}) and (\ref{CD}). The derivation will also require an equivalent of (\ref{ABCA}) and we shall use 
the following one,
\[
\label{DBCD}
\gen{\pi}_{n+1}^{\mathcal{C}} \left( \tilde{\mathcal{D}}(\lambda_0) [\gen{Y}^{1,n}]_{\mathcal{B}} \right) =  \gen{\pi}_{n+1}^{\mathcal{B}} \left( [\gen{X}^{1,n}]_{\mathcal{C}} \; \tilde{\mathcal{D}}(\lambda_0) \right) \; .
\]

\begin{theorem} \label{funD}
Let $\widetilde{M}_0$, $\widetilde{N}_{\lambda}^{(\mathcal{C})}$ and $\widetilde{N}_{\lambda}^{(\mathcal{B})}$ be defined as 
\<
\label{coeffD}
\widetilde{M}_0 &\coloneqq& \Lambda_{\tilde{\mathcal{D}}} (\lambda_0) \left\{ \prod_{\lambda \in \gen{Y}^{1,n}} \frac{a(\lambda_0 - \lambda)}{b(\lambda_0 - \lambda)} \frac{a(\lambda_0 + \lambda + \gamma)}{b(\lambda_0 + \lambda + \gamma)} - \prod_{\lambda \in \gen{X}^{1,n}} \frac{a(\lambda_0 - \lambda)}{b(\lambda_0 - \lambda)} \frac{a(\lambda_0 + \lambda + \gamma)}{b(\lambda_0 + \lambda + \gamma)} \right\} \nonumber \\
\widetilde{N}_{\lambda}^{(\mathcal{C})} &\coloneqq& - \Lambda_{\mathcal{A}} (\lambda) \frac{c(2 \lambda_0)}{a(2 \lambda_0)} \frac{b(2 \lambda)}{a(2 \lambda)} \frac{a(2 \lambda_0 + \gamma)}{a(\lambda_0 + \lambda)} \prod_{\tilde{\lambda} \in \gen{X}_{\lambda}^{1,n}} \frac{a(\tilde{\lambda} - \lambda)}{b(\tilde{\lambda} - \lambda)} \frac{b(\tilde{\lambda} + \lambda)}{a(\tilde{\lambda} + \lambda)} \nonumber \\
&& \qquad + \; \Lambda_{\tilde{\mathcal{D}}} (\lambda) \frac{a(2 \lambda_0 + \gamma)}{b(2 \lambda_0 + \gamma)} \frac{c(\lambda_0 - \lambda)}{b(\lambda_0 - \lambda)} \prod_{\tilde{\lambda} \in \gen{X}_{\lambda}^{1,n}} \frac{a(\lambda - \tilde{\lambda})}{b(\lambda - \tilde{\lambda})} \frac{a(\lambda + \tilde{\lambda} + \gamma)}{b(\lambda + \tilde{\lambda} + \gamma)} \nonumber \\
\widetilde{N}_{\lambda}^{(\mathcal{B})} &\coloneqq& \Lambda_{\mathcal{A}} (\lambda) \frac{c(2 \lambda_0)}{a(2 \lambda_0)} \frac{b(2 \lambda)}{a(2 \lambda)} \frac{a(2 \lambda_0 + \gamma)}{a(\lambda_0 + \lambda)} \prod_{\tilde{\lambda} \in \gen{Y}_{\lambda}^{1,n}} \frac{a(\tilde{\lambda} - \lambda)}{b(\tilde{\lambda} - \lambda)} \frac{b(\tilde{\lambda} + \lambda)}{a(\tilde{\lambda} + \lambda)} \nonumber \\
&& \qquad - \; \Lambda_{\tilde{\mathcal{D}}}  (\lambda) \frac{a(2 \lambda_0 + \gamma)}{b(2 \lambda_0 + \gamma)} \frac{c(\lambda_0 - \lambda)}{b(\lambda_0 - \lambda)} \prod_{\tilde{\lambda} \in \gen{Y}_{\lambda}^{1,n}} \frac{a(\lambda - \tilde{\lambda})}{b(\lambda - \tilde{\lambda})} \frac{a(\lambda + \tilde{\lambda} + \gamma)}{b(\lambda + \tilde{\lambda} + \gamma)} \; . \nonumber \\
\>
Then the functional equation 
\[
\label{typeD}
\widetilde{M}_0 \; \mathcal{S}_n (\vec{X}^{1,n} | \vec{Y}^{1,n}) + \sum_{\lambda \in \gen{X}^{1,n}} \widetilde{N}_{\lambda}^{(\mathcal{C})} \; \mathcal{S}_n (\vec{X}_{\lambda}^{0,n} | \vec{Y}^{1,n}) 
+ \sum_{\lambda \in \gen{Y}^{1,n}} \widetilde{N}_{\lambda}^{(\mathcal{B})} \; \mathcal{S}_n (\vec{X}^{1,n} | \vec{Y}_{\lambda}^{0,n}) = 0 
\]
is satisfied by the scalar product $\mathcal{S}_n$.
\end{theorem}
\begin{proof}
We follow the same strategy used in the proof of Theorem \ref{funA}. The first step is to apply the map
$\gen{\pi}_{n+1}^{\mathcal{C}}$ on the relation (\ref{DB}). By doing so we are left with terms $\gen{\pi}_{n+1}^{\mathcal{C}} \left( [\gen{Y}^{1,n}]_{\mathcal{B}} \tilde{\mathcal{D}}(\lambda_0) \right)$, $\gen{\pi}_{n+1}^{\mathcal{C}} \left( [\gen{Y}_{\lambda}^{0,n}]_{\mathcal{B}} \mathcal{A}(\lambda) \right)$
and $\gen{\pi}_{n+1}^{\mathcal{C}} \left( [\gen{Y}_{\lambda}^{0,n}]_{\mathcal{B}} \tilde{\mathcal{D}}(\lambda) \right)$, which obey a reduction
relation $\gen{\pi}_{n+1}^{\mathcal{C}} \to  \gen{\pi}_{n}^{\mathcal{C}}$ due to the $\alg{sl}(2)$ highest-weight property of (\ref{zero}). Next we apply
$\gen{\pi}_{n+1}^{\mathcal{B}}$ to (\ref{CD}). The terms obtained through this procedure are of the following form:
$\gen{\pi}_{n+1}^{\mathcal{B}} \left( \tilde{\mathcal{D}}(\lambda_0) \; [\gen{X}^{1,n}]_{\mathcal{C}} \right)$, 
$\gen{\pi}_{n+1}^{\mathcal{B}} \left( \mathcal{A}(\lambda) \; [\gen{X}_{\lambda}^{0,n}]_{\mathcal{C}} \right)$ and $\gen{\pi}_{n+1}^{\mathcal{B}} \left( \tilde{\mathcal{D}}(\lambda) \; [\gen{X}_{\lambda}^{0,n}]_{\mathcal{C}} \right)$.
The reduction property  $\gen{\pi}_{n+1}^{\mathcal{B}, \mathcal{C}} \to  \gen{\pi}_{n}^{\mathcal{B},\mathcal{C}}$ then allows one to identify the scalar product $\mathcal{S}_n$
through (\ref{piS}). The mechanism above described then yields the functional relation (\ref{typeD}). 
\end{proof}

\begin{remark} \label{rema}
At this stage it is important to remark some structural properties of (\ref{typeA}) and (\ref{typeD}). Although
we can readily see that those equations are invariant under the permutations of variables
$\lambda_i^{\mathcal{B}} \leftrightarrow \lambda_j^{\mathcal{B}}$ and $\lambda_i^{\mathcal{C}} \leftrightarrow \lambda_j^{\mathcal{C}}$
for any $i,j \in \{1, 2 , \dots , n\}$, the same does not hold for the permutations $\lambda_0 \leftrightarrow \lambda_i^{\mathcal{B}}$
and $\lambda_0 \leftrightarrow \lambda_i^{\mathcal{C}}$. Therefore, the latter permutations indeed produces significantly more
independent equations describing the scalar product $\mathcal{S}_n$.
\end{remark}

\section{The scalar product $\mathcal{S}_n$}
\label{sec:SOL}

In the previous section we have derived two functional equations satisfied by the scalar product $\mathcal{S}_n$.
This set of equations is described in Theorems \ref{funA} and \ref{funD}, and their derivation makes use of higher-order reflection 
algebra relations and the highest-weight property of the vector $\ket{0}$ entering in the definition (\ref{scp}).
However, it is important to stress here that equations (\ref{typeA}) and (\ref{typeD}) carry no information about the particular
representation of the operators $\mathcal{B}$'s and $\mathcal{C}$'s we are considering in the definition of the scalar 
product (\ref{scp}). The choice of representation characterizes the particular model for which the scalar products are being computed. 
In this way, it is expected that different choices of representations for the operators $\mathcal{B}$'s and $\mathcal{C}$'s will be manifested
in different classes of solutions of the system (\ref{typeA}, \ref{typeD}). Thus, we need to characterize
the class of solutions of (\ref{typeA}, \ref{typeD}) we are interested. 
The class of solution corresponding to the scalar product of Bethe vectors associated with the $XXZ$ chain with open 
boundaries (\ref{scp}) will be discussed in what follows, but some remarks are in order before proceeding with that.

For instance, the system of equations (\ref{typeA}, \ref{typeD}) exhibits the same structure as the one
derived in \cite{Galleas_SCP} for the scalar product of Bethe vectors associated with the $XXZ$ chain with periodic
boundary conditions. The main difference between our present equations (\ref{typeA}, \ref{typeD}) and the equations of
\cite{Galleas_SCP} lies on the particular form of the equations coefficients. The steps required for solving 
 (\ref{typeA}, \ref{typeD}) need to be modified due to that difference but the general idea remains the same.
In fact, the general strategy for solving this type of functional relations was firstly
described in \cite{Galleas_2012, Galleas_2013}. The resolution of (\ref{typeA}, \ref{typeD}) will follow a number of
systematic steps exploiting properties expected for the desired scalar product (\ref{scp}).

\begin{lemma}[Doubly symmetric function] \label{symmetry}
Let $\mathcal{Z} \coloneqq (\lambda_1 , \lambda_2 , \dots , \lambda_n)$
be a generic $n$-tuple. Then the scalar product $\mathcal{S}_n$ satisfies the symmetry property, 
\<
\label{symC}
\mathcal{S}_n (\lambda_1^{\mathcal{C}}, \dots , \lambda_n^{\mathcal{C}} | \lambda_1^{\mathcal{B}}, \dots , \lambda_n^{\mathcal{B}} ) = \mathcal{S}_n (\mathcal{Z}  | \lambda_1^{\mathcal{B}}, \dots , \lambda_n^{\mathcal{B}} )  \qquad \forall \; \mathcal{Z} \in \gen{Sym}( \{ \lambda_1^{\mathcal{C}}, \lambda_2^{\mathcal{C}} , \dots , \lambda_n^{\mathcal{C}} \} ) \; . \nonumber \\
\>
In addition to that we also have the property
\<
\label{symB}
\mathcal{S}_n (\lambda_1^{\mathcal{C}}, \dots , \lambda_n^{\mathcal{C}} | \lambda_1^{\mathcal{B}}, \dots , \lambda_n^{\mathcal{B}} ) = \mathcal{S}_n ( \lambda_1^{\mathcal{C}}, \dots , \lambda_n^{\mathcal{C}} |   \mathcal{Z} )  \qquad \forall \; \mathcal{Z} \in \gen{Sym}( \{ \lambda_1^{\mathcal{B}}, \lambda_2^{\mathcal{B}} , \dots , \lambda_n^{\mathcal{B}} \} ) \; . \nonumber \\
\>

\end{lemma}
\begin{proof}
See \Appref{sec:SYM}. 
\end{proof}

\begin{remark}
Due to Lemma \ref{symmetry} we can also consider the simplified notation 
\[
\mathcal{S}_n (\lambda_1^{\mathcal{C}}, \dots , \lambda_n^{\mathcal{C}} | \lambda_1^{\mathcal{B}}, \dots , \lambda_n^{\mathcal{B}} ) = \mathcal{S}_n ( \gen{X}^{1,n} | \gen{Y}^{1,n} ) \; .
\]
\end{remark}

\begin{lemma}[Polynomial structure] \label{polynomial}
Let $x_i^{\mathcal{B}, \mathcal{C}} \coloneqq e^{2 \lambda_i^{\mathcal{B} , \mathcal{C}}}$. Then $\mathcal{S}_n$ is of the form,
\[
\label{poly}
\mathcal{S}_n ( \gen{X}^{1,n} | \gen{Y}^{1,n} ) = \bar{\mathcal{S}_n} ( x_1^{\mathcal{C}}, \dots , x_n^{\mathcal{C}} | x_1^{\mathcal{B}}, \dots , x_n^{\mathcal{B}}) \prod_{i=1}^n \left( x_i^{\mathcal{B}} x_i^{\mathcal{C}} \right)^{-L} \; ,
\]
where $\bar{\mathcal{S}_n}$ is a doubly symmetric polynomial of order $2L$ in each one of its variables.
\end{lemma}
\begin{proof}
See \Appref{sec:POL}. 
\end{proof}

\begin{lemma}[Special zeroes] \label{zeroes} 
The function $\mathcal{S}_n (\gen{X}^{1,n} | \gen{Y}^{1,n} )$ vanishes for the specialization of variables
\[ \label{zeroesB}
\left( \lambda_1^{\mathcal{B}} , \lambda_2^{\mathcal{B}}   \right) \in \left\{ \left( \mu_1 - \gamma , \mu_1  \right), \left( \mu_1 - \gamma , -\mu_1 - \gamma  \right), \left( -\mu_1 , \mu_1  \right)  \right\} \; .
\]
Similarly, it also vanishes for 
\[ \label{zeroesC}
\left( \lambda_1^{\mathcal{C}} , \lambda_2^{\mathcal{C}}   \right) \in \left\{ \left( \mu_1 - \gamma , \mu_1  \right), \left( \mu_1 - \gamma , -\mu_1 - \gamma  \right), \left( -\mu_1 , \mu_1  \right)  \right\} \; .
\] 
\end{lemma}
\begin{proof}
See \Appref{sec:ZEROES}. 
\end{proof}

\begin{lemma}[Asymptotic behavior] \label{asymptotic}
In the limit $x_i^{\mathcal{B}, \mathcal{C}} \to \infty \; \colon \; \forall i \in \{1,2, \dots , n\}$, the scalar product
(\ref{scp}) exhibits the following asymptotic behavior,
\<
\label{asymp}
\bar{\mathcal{S}}_n &\sim& \frac{q^{2n (L-1)}}{2^{2n(2L+1)}} (q - q^{-1})^{2n} \prod_{i=1}^{n} \left( x_i^{\mathcal{B}} x_i^{\mathcal{C}} \right)^{2L} \nonumber \\
&& \times \sum_{1 \leq r_i < r_{i+1}  \leq L} \prod_{s=1}^{n} \prod_{\epsilon \in \{ \pm \}} \left( t y_{r_s}^{- \epsilon \frac{1}{2}} q^{L - r_s} \Delta_{n-s}^{\epsilon} - t^{-1} y_{r_s}^{\epsilon \frac{1}{2}} q^{r_s - L} \Delta_{n-s}^{-\epsilon} \right) \; , \nonumber \\
\>
where $q \coloneqq e^{\gamma}$, $t \coloneqq e^{h}$, $y_i \coloneqq e^{2\mu_i}$ and $\Delta_{m}^{\pm} \coloneqq \sum_{l=0}^m q^{\pm 2l}$.
\end{lemma}
\begin{proof}
See \Appref{sec:ASYMP}. 
\end{proof}

The above lemmas pave the way for the resolution of the system of equations (\ref{typeA}, \ref{typeD}).
The desired solution is given by the following theorem.

\begin{theorem} \label{off-shell}
The scalar product $\mathcal{S}_n$ can be written as the following multiple contour integral,
\< \label{sol}
&& \mathcal{S}_n (\gen{X}^{1,n} | \gen{Y}^{1,n} ) = \nonumber \\
&& c^{2 n} \oint \dots \oint   \frac{\prod_{1 \leq i < j \leq n} a(w_j - \mu_i) b(w_j + \mu_i) a(\bar{w}_j - \mu_i) b(\bar{w}_j + \mu_i)}{\prod_{i,j=1}^n b(w_i - \lambda_j^{C}) b(\bar{w}_i - \lambda_j^{B})} \prod_{i=1}^n \frac{b(2 \mu_i)}{a(2 \mu_i)} \nonumber \\
&& \qquad \quad \prod_{1 \leq i < j \leq n} b(w_j - w_i)^2 b(\bar{w}_j - \bar{w}_i)^2 \prod_{i=1}^n \frac{b(2 w_i)}{a(2 w_i)} \frac{b(2 \bar{w}_i)}{a(2 \bar{w}_i)} R_i^{-1} \; \gen{det} (\gen{\Phi}^{(i)}) \prod_{k=1}^n \frac{\dd w_k}{2 \ii \pi} \frac{\dd \bar{w}_k}{2 \ii \pi} \nonumber \\
\>
where 
\< 
\label{RI}
R_i &\coloneqq& \prod_{k=i}^n \frac{a(w_k - \mu_i)}{b(w_k - \mu_i)} \frac{b(w_k + \mu_i)}{a(w_k + \mu_i)} - \prod_{k=i}^n \frac{a(\bar{w}_k - \mu_i)}{b(\bar{w}_k - \mu_i)} \frac{b(\bar{w}_k + \mu_i)}{a(\bar{w}_k + \mu_i)} 
\>
and $\gen{\Phi}^{(i)}$ is a $2 \times 2$ matrix with entries
\<
\label{PhiI}
&& \gen{\Phi}^{(i)}_{lm} \coloneqq \nonumber \\
&& \frac{b(h + s_i^m)}{\omega_l (s_i^m - (-1)^{l-1} \mu_i)} \prod_{k=i}^L a(s_i^m - \mu_k) a(s_i^m + \mu_k) \prod_{j=i+1}^n \frac{a(s_j^m - s_i^m)}{b(s_j^m - s_i^m)} \frac{b(s_j^m + s_i^m)}{a(s_j^m + s_i^m)} \nonumber \\
&& - \frac{a(s_i^m - h)}{\bar{\omega}_l (s_i^m + (-1)^{l-1} \mu_i)} \prod_{k=i}^L b(s_i^m - \mu_k) b(s_i^m + \mu_k) \prod_{j=i+1}^n \frac{a(s_i^m - s_j^m)}{b(s_i^m - s_j^m)} \frac{a(s_i^m + s_j^m + \gamma)}{b(s_i^m + s_j^m + \gamma)} \; . \nonumber \\
\>
Formula (\ref{PhiI}) takes into account the conventions $\omega_l (x) \coloneqq \sinh{(x + (l-1)\gamma)}$,
$\bar{\omega}_l (x) \coloneqq \sinh{(x + (2-l)\gamma)}$ and $s_i^m \coloneqq \delta_{1m} w_i + \delta_{2m} \bar{w}_i$.
The variables $w_i$ and $\bar{w}_i$ in (\ref{sol}) are auxiliary integration variables and the integration
contour for each variable $w_i$ encloses only all points in the set $\gen{X}^{1,n}$. Similarly, each variable $\bar{w}_i$ is
integrated along a contour containing only all points in the set $\gen{Y}^{1,n}$.  
\end{theorem}
\begin{proof}
The proof follows from the resolution of the system of equations (\ref{typeA}, \ref{typeD}) under the conditions
established by Lemmas \ref{polynomial} and \ref{asymptotic}. In fact, these two lemmas are the only ones that do not
follow from the functional equations (\ref{typeA}) and (\ref{typeD}), and they characterize the class of solution 
we are interested. Our procedure will be described by a sequence of systematic steps.

\vspace{0.3cm}
\noindent \textit{Step $1$.} Let us introduce the notation $\bar{\gen{Z}}^{i,j} \coloneqq \gen{Z}^{i,j} \cup \{ \mu_1 - \gamma \}$
and $\check{\gen{Z}}^{i,j} \coloneqq \gen{Z}^{i,j} \cup \{ \mu_1 \}$ for $\gen{Z}^{i,j} \in \{ \gen{X}^{i,j} , \gen{Y}^{i,j} \}$.
Thus, $\bar{\gen{Z}}^{i,j}$ and $\check{\gen{Z}}^{i,j}$ are sets of cardinalities $j-i+1$, and due to  Lemmas 
\ref{polynomial}, \ref{zeroes} and \ref{symmetry} we can write
\<
\label{SV}
\mathcal{S}_n ( \bar{\gen{X}}^{2,n} | \check{\gen{Y}}^{2,n} ) = \prod_{\lambda \in \gen{X}^{2,n} } b(\lambda - \mu_1) a(\lambda + \mu_1) \prod_{\bar{\lambda} \in \gen{Y}^{2,n} } a(\bar{\lambda} - \mu_1) b(\bar{\lambda} + \mu_1) \; V( \gen{X}^{2,n} | \gen{Y}^{2,n} ) \nonumber \\
\>
and
\<
\label{SW}
\mathcal{S}_n ( \check{\gen{X}}^{2,n} | \bar{\gen{Y}}^{2,n} ) = \prod_{\lambda \in \gen{X}^{2,n} } a(\lambda - \mu_1) b(\lambda + \mu_1) \prod_{\bar{\lambda} \in \gen{Y}^{2,n} } b(\bar{\lambda} - \mu_1) a(\bar{\lambda} + \mu_1) \; W( \gen{X}^{2,n} | \gen{Y}^{2,n} ) \; . \nonumber \\ 
\>
The functions $V$ and $W$ in (\ref{SV}) and (\ref{SW}) must have the same polynomial structure of $\mathcal{S}_n$, as described in (\ref{poly}), under
the mappings $L \mapsto L-1$ and $n \mapsto n-1$.

\vspace{0.3cm}
\noindent \textit{Step $2$.} Set $\lambda_0 = \mu_1 - \gamma$ and $\lambda_n^{B} = \mu_1$ in Eq. (\ref{typeA}).
By doing so we can readily make use of (\ref{zeroesB}) and (\ref{SV}). We are then left with the relation
\[
\label{FYOM}
\mathcal{S}_n (\gen{X}^{1,n} | \bar{\gen{Y}}^{1,n-1} ) = \mathcal{F}(\gen{Y}^{1,n-1}) \sum_{\lambda \in \gen{X}^{1,n}} \gen{\Omega}_{\lambda}(\gen{X}_{\lambda}^{1,n}) \; V(\gen{X}_{\lambda}^{1,n} | \gen{Y}^{1,n-1}) \; ,
\]
where
\<
\label{FOM}
\mathcal{F}(\gen{Z}^{1,n-1}) &\coloneqq& \left[ c \; b(2 \mu_1) b(h + \mu_1) \prod_{j=2}^L a(\mu_1 - \mu_j) a(\mu_1 + \mu_j) \right]^{-1} \prod_{\lambda \in \gen{Z}^{1,n-1}} b(\lambda - \mu_1) a(\lambda + \mu_1) \nonumber \\
\gen{\Omega}_{\lambda}(\gen{Z}_{\lambda}^{1,n}) &\coloneqq& \Lambda_{\mathcal{A}}(\lambda) \frac{c(\lambda - \mu_1)}{a(\lambda - \mu_1)} \frac{b(2\lambda)}{a(2\lambda)} \prod_{\bar{\lambda} \in \gen{Z}_{\lambda}^{1,n} } \frac{a(\bar{\lambda} - \lambda)}{b(\bar{\lambda} - \lambda)} \frac{b(\bar{\lambda} + \lambda)}{a(\bar{\lambda} + \lambda)} b(\bar{\lambda} - \mu_1) a(\bar{\lambda} + \mu_1) \nonumber \\
&& + \; \Lambda_{\tilde{\mathcal{D}}}(\lambda) \frac{c(\lambda + \mu_1)}{a(\lambda + \mu_1)} \prod_{\bar{\lambda} \in \gen{Z}_{\lambda}^{1,n} } 
\frac{a(\lambda - \bar{\lambda})}{b(\lambda - \bar{\lambda})} \frac{a(\bar{\lambda} + \lambda + \gamma)}{b(\bar{\lambda} + \lambda + \gamma)} b(\bar{\lambda} - \mu_1) a(\bar{\lambda} + \mu_1) \; . \nonumber \\
\>

\vspace{0.3cm}
\noindent \textit{Step $3$.} Set $\lambda_0 = \mu_1 - \gamma$ and $\lambda_n^{C} = \mu_1$ in Eq. (\ref{typeA}).
This particular specialization of variables paves the way for using (\ref{zeroesC}) and (\ref{SW}). We then obtain
the expression
\[
\label{FXOM}
\mathcal{S}_n (\bar{\gen{X}}^{1,n-1}  | \gen{Y}^{1,n} ) = \mathcal{F}(\gen{X}^{1,n-1}) \sum_{\lambda \in \gen{Y}^{1,n}} \gen{\Omega}_{\lambda}(\gen{Y}_{\lambda}^{1,n}) \; W(\gen{X}^{1,n-1} | \gen{Y}_{\lambda}^{1,n} ) \; ,
\]
with functions $\mathcal{F}$ and $\gen{\Omega}_{\lambda}$ previously defined in (\ref{FOM}).

\vspace{0.3cm}
\noindent \textit{Step $4$.} Now we draw our attention to Eq. (\ref{typeD}). Then we set $\lambda_0 = \mu_1$
and $\lambda_n^{B} = \mu_1 - \gamma$ in (\ref{typeD}) using (\ref{zeroesB}) and (\ref{SW}). This procedure
yields the formula
\[
\label{bFYOM}
\mathcal{S}_n ( \gen{X}^{1,n} | \check{\gen{Y}}^{1,n-1} ) = \bar{\mathcal{F}}(\gen{Y}^{1,n-1}) \sum_{\lambda \in \gen{X}^{1,n}} \bar{\gen{\Omega}}_{\lambda}(\gen{X}_{\lambda}^{1,n}) \; W(\gen{X}_{\lambda}^{1,n} | \gen{Y}^{1,n-1} ) \; ,
\]
where
\<
\label{bFOM}
\bar{\mathcal{F}}(\gen{Z}^{1,n-1}) &\coloneqq& \left[ c \; b(2 \mu_1 - 2 \gamma) b(h - \mu_1) \prod_{j=2}^L a(\mu_j - \mu_1) a(-\mu_j - \mu_1) \right]^{-1} \nonumber \\
&& \times \prod_{\lambda \in \gen{Z}^{1,n-1}} a(\lambda - \mu_1) b(\lambda + \mu_1) \nonumber \\
\bar{\gen{\Omega}}_{\lambda}(\gen{Z}_{\lambda}^{1,n}) &\coloneqq& \Lambda_{\mathcal{A}}(\lambda) \frac{c(\lambda + \mu_1)}{a(\lambda + \mu_1)} \frac{b(2\lambda)}{a(2\lambda)} \prod_{\bar{\lambda} \in \gen{Z}_{\lambda}^{1,n} } \frac{a(\bar{\lambda} - \lambda)}{b(\bar{\lambda} - \lambda)} \frac{b(\bar{\lambda} + \lambda)}{a(\bar{\lambda} + \lambda)} a(\bar{\lambda} - \mu_1) b(\bar{\lambda} + \mu_1) \nonumber \\
&& + \; \Lambda_{\tilde{\mathcal{D}}}(\lambda) \frac{c(\lambda - \mu_1)}{b(\lambda - \mu_1)} \prod_{\bar{\lambda} \in \gen{Z}_{\lambda}^{1,n} } \frac{a(\lambda - \bar{\lambda})}{b(\lambda - \bar{\lambda})} \frac{a(\bar{\lambda} + \lambda + \gamma)}{b(\bar{\lambda} + \lambda + \gamma)} a(\bar{\lambda} - \mu_1) b(\bar{\lambda} + \mu_1) \; . \nonumber \\
\>

\vspace{0.3cm}
\noindent \textit{Step $5$.} The next step consists in looking at (\ref{typeD}) under the specialization
$\lambda_0 = \mu_1$ and $\lambda_n^{C} = \mu_1 - \gamma$. The resulting expression can be simplified with the help
of (\ref{zeroesC}) and (\ref{SV}). We are thus left with the relation
\[
\label{bFXOM}
\mathcal{S}_n ( \check{\gen{X}}^{1,n-1} |  \gen{Y}^{1,n} ) = \bar{\mathcal{F}}(\gen{X}^{1,n-1}) \sum_{\lambda \in \gen{Y}^{1,n}} \bar{\gen{\Omega}}_{\lambda}(\gen{Y}_{\lambda}^{1,n}) \; V(\gen{X}^{1,n-1} | \gen{Y}_{\lambda}^{1,n} ) \; ,
\]
where $\bar{\mathcal{F}}$ and $\bar{\gen{\Omega}}_{\lambda}$ have been defined in (\ref{bFOM}).

\vspace{0.3cm}
\noindent \textit{Step $6$.} Next we set $\lambda_n^{C} = \mu_1$ in (\ref{FYOM}) and compare the result with (\ref{SW}). This allows us 
to conclude that $V = W$. The same conclusion can be obtained by setting $\lambda_n^{B} = \mu_1$ in (\ref{FXOM}) and comparing
the result with (\ref{SV}). It is worth remarking that a similar analysis using the results of \textit{Steps} $4$ and $5$ also
yields the same constraint.

\vspace{0.3cm}
\noindent \textit{Step $7$.} We set $\lambda_0 = \mu_1$ in Eq. (\ref{typeA}) and use the relations (\ref{bFYOM}) and  
(\ref{bFXOM}). From this point on, we are already assuming $V = W$ as obtained in \textit{Step $6$}. 
This procedure allows us to write
\[
\label{SXY}
\mathcal{S}_n (\gen{X}^{1,n} | \gen{Y}^{1,n}) = \sum_{\lambda \in \gen{X}^{1,n}} \sum_{\bar{\lambda} \in \gen{Y}^{1,n}} \mathcal{K}_{\lambda \bar{\lambda}} \; V(\gen{X}_{\lambda}^{1,n} | \gen{Y}_{\bar{\lambda}}^{1,n} ) \; .
\]
The function $\mathcal{K}_{\lambda \bar{\lambda}}$ is given by
\<
\label{KLbL}
\mathcal{K}_{\lambda \bar{\lambda}} \coloneqq \mathcal{C}_0 \; \gen{det} (\gen{\Gamma}_{\lambda \bar{\lambda}}) \frac{b(2 \lambda)}{a(2 \lambda)} \frac{b(2 \bar{\lambda})}{a(2 \bar{\lambda})} \prod_{\tilde{\lambda} \in \gen{X}_{\lambda}^{1,n} \cup \gen{Y}_{\bar{\lambda}}^{1,n}} a(\tilde{\lambda} - \mu_1) b(\tilde{\lambda} + \mu_1)  \; ,
\>
where
\<
\mathcal{C}_0^{-1} &\coloneqq&   b(h + \mu_1) b(h - \mu_1) b(2 \mu_1 - 2\gamma) b(2 \mu_1 + \gamma) \prod_{j=2}^L \prod_{\epsilon = \pm 1} a(\mu_1 - \epsilon \mu_j) a(\epsilon \mu_j - \mu_1) \nonumber \\
&& \times \left[ \prod_{\tilde{\lambda} \in \gen{X}^{1,n}} \frac{a(\tilde{\lambda} - \mu_1)}{b(\tilde{\lambda} - \mu_1)} \frac{b(\tilde{\lambda} + \mu_1)}{a(\tilde{\lambda} + \mu_1)} -
\prod_{\tilde{\lambda} \in \gen{Y}^{1,n}} \frac{a(\tilde{\lambda} - \mu_1)}{b(\tilde{\lambda} - \mu_1)} \frac{b(\tilde{\lambda} + \mu_1)}{a(\tilde{\lambda} + \mu_1)} \right]
\>
is a constant, i.e. independent of $\lambda$ and $\bar{\lambda}$, and $\gen{\Gamma}_{\lambda \bar{\lambda}}$ is a $2 \times 2$ matrix
with entries
\<
\label{gama}
(\gen{\Gamma}_{\lambda \bar{\lambda}})_{ij} &\coloneqq&  \frac{b(h+\kappa_j)}{\omega_i (\kappa_j - (-1)^{i-1} \mu_1)} \prod_{k=1}^L a(\kappa_j - \mu_k) a(\kappa_j + \mu_k) \prod_{\tilde{\lambda} \in (\gen{Z}_i)_{\lambda}^{1,n} } \frac{a(\tilde{\lambda} - \kappa_j)}{b(\tilde{\lambda} - \kappa_j)} \frac{b(\tilde{\lambda} + \kappa_j)}{a(\tilde{\lambda} + \kappa_j)}  \nonumber \\
&-&  \frac{a(\kappa_j - h)}{\bar{\omega}_i (\kappa_j + (-1)^{i-1} \mu_1)} \prod_{k=1}^L b(\kappa_j - \mu_k) b(\kappa_j + \mu_k)  \prod_{\tilde{\lambda} \in (\gen{Z}_i)_{\lambda}^{1,n} } \frac{a(\kappa_j - \tilde{\lambda})}{b(\kappa_j - \tilde{\lambda})} \frac{a(\tilde{\lambda} + \kappa_j + \gamma)}{b(\tilde{\lambda} + \kappa_j + \gamma)} \; . \nonumber \\
\>
In (\ref{gama}) we have employed the notation $\omega_i (x) \coloneqq \sinh{(x + (i-1)\gamma)}$ and 
$\bar{\omega}_i (x) \coloneqq \sinh{(x + (2-i)\gamma)}$. Formula (\ref{gama}) also considers 
$\kappa_j \coloneqq \delta_{1j} \lambda + \delta_{2j} \bar{\lambda}$ and
\[
\label{KZ}
\gen{Z}_j \coloneqq \begin{cases}
\gen{X} \quad \mbox{for} \; j=1 \cr
\gen{Y} \quad \mbox{for} \; j=2 \cr
\end{cases} \; .
\]

\vspace{0.3cm}
\noindent \textit{Step $8$.} Next we look at Eq. (\ref{typeD}) under the specialization $\lambda_0 = \mu_1 - \gamma$.
This particular specialization allows us to readily use formulae (\ref{FYOM}) and (\ref{FXOM}). By doing so we
obtain the relation
\[
\label{bSXY}
\mathcal{S}_n (\gen{X}^{1,n} | \gen{Y}^{1,n}) = \sum_{\lambda \in \gen{X}^{1,n}} \sum_{\bar{\lambda} \in \gen{Y}^{1,n}} \bar{\mathcal{K}}_{\lambda \bar{\lambda}} \; V(\gen{X}_{\lambda}^{1,n} | \gen{Y}_{\bar{\lambda}}^{1,n} ) \; ,
\]
where 
\<
\label{bKLbL}
\bar{\mathcal{K}}_{\lambda \bar{\lambda}} \coloneqq \bar{\mathcal{C}}_0 \; \gen{det} (\bar{\gen{\Gamma}}_{\lambda \bar{\lambda}}) \frac{b(2 \lambda)}{a(2 \lambda)} \frac{b(2 \bar{\lambda})}{a(2 \bar{\lambda})}  \prod_{\tilde{\lambda} \in \gen{X}_{\lambda}^{1,n} \cup \gen{Y}_{\bar{\lambda}}^{1,n}} b(\tilde{\lambda} - \mu_1) a(\tilde{\lambda} + \mu_1)  \; .
\>
The term $\bar{\mathcal{C}}_0$ is a constant and it explicitly reads
\<
\bar{\mathcal{C}}_0^{-1} &\coloneqq&  b(h + \mu_1) b(h - \mu_1) b(2 \mu_1 - 2\gamma) b(2 \mu_1 - \gamma) \prod_{j=2}^L \prod_{\epsilon = \pm 1} a(\mu_1 - \epsilon \mu_j) a(\epsilon \mu_j - \mu_1) \nonumber \\
&& \times \left[ \prod_{\tilde{\lambda} \in \gen{X}^{1,n}} \frac{b(\tilde{\lambda} - \mu_1)}{a(\tilde{\lambda} - \mu_1)} \frac{a(\tilde{\lambda} + \mu_1)}{b(\tilde{\lambda} + \mu_1)} -
\prod_{\tilde{\lambda} \in \gen{Y}^{1,n}} \frac{b(\tilde{\lambda} - \mu_1)}{a(\tilde{\lambda} - \mu_1)} \frac{a(\tilde{\lambda} + \mu_1)}{b(\tilde{\lambda} + \mu_1)} \right] \; .
\>
Here $\bar{\gen{\Gamma}}_{\lambda \bar{\lambda}}$ is also a $2 \times 2$ matrix with entries given by
\<
\label{bgama}
(\bar{\gen{\Gamma}}_{\lambda \bar{\lambda}})_{ij} &\coloneqq&  \frac{b(h+\kappa_j)}{\omega_i (\kappa_j + (-1)^{i-1} \mu_1)} \prod_{k=1}^L a(\kappa_j - \mu_k) a(\kappa_j + \mu_k) \prod_{\tilde{\lambda} \in (\gen{Z}_i)_{\lambda}^{1,n} } \frac{a(\tilde{\lambda} - \kappa_j)}{b(\tilde{\lambda} - \kappa_j)} \frac{b(\tilde{\lambda} + \kappa_j)}{a(\tilde{\lambda} + \kappa_j)}  \nonumber \\
&-&  \frac{a(\kappa_j - h)}{\bar{\omega}_i (\kappa_j - (-1)^{i-1} \mu_1)} \prod_{k=1}^L b(\kappa_j - \mu_k) b(\kappa_j + \mu_k)  \prod_{\tilde{\lambda} \in (\gen{Z}_i)_{\lambda}^{1,n} } \frac{a(\kappa_j - \tilde{\lambda})}{b(\kappa_j - \tilde{\lambda})} \frac{a(\tilde{\lambda} + \kappa_j + \gamma)}{b(\tilde{\lambda} + \kappa_j + \gamma)} \; . \nonumber \\
\>

\vspace{0.3cm}
\noindent \textit{Step $9$.} Substitute formula (\ref{SXY}) in Eq. (\ref{typeA}) and set $\lambda_n^{B} = \mu_1 - \gamma$
and $\lambda_n^{C} = \mu_1$. This yields an equation for the function $V$ exhibiting the same structure of (\ref{typeA}), but
with modified coefficients and $n \mapsto n-1$. The explicit form of the coefficients will not be relevant for our purposes
here as one can verify that the resulting equation consists of a linear combination of (\ref{typeA}) and (\ref{typeD})
under the mappings $L \mapsto L-1$, $n \mapsto n-1$ and $\mu_i \mapsto \mu_{i+1}$.

\vspace{0.3cm}
\noindent \textit{Step $10$.} We repeat the procedure of \textit{Step $9$} but substitute instead formula
(\ref{SXY}) in Eq. (\ref{typeD}). We then set $\lambda_n^{C} = \mu_1 - \gamma$ and $\lambda_n^{B} = \mu_1$
to find another equation for the function $V$ with the same structure of (\ref{typeD}). 
Similarly to the previous step, one can verify that the resulting equation consists of a linear combination of (\ref{typeA}) and (\ref{typeD})
taking into account the mappings $L \mapsto L-1$, $n \mapsto n-1$ and $\mu_i \mapsto \mu_{i+1}$.
\begin{figure} \centering
\scalebox{1}{
\begin{tikzpicture}[>=stealth]
\path (0,0) node[rectangle,fill=gray!20!white,draw] (p1) {Equation type A}
      (10,0) node[rectangle,fill=gray!20!white,draw] (p2) {Equation type D};
\begin{scope}[yshift=0cm]
\path (2,2) node[rectangle,rounded corners=7pt,fill=gray!20!white,draw] (p3) {Lemma $2$}
      (8,2) node[rectangle,rounded corners=7pt,fill=gray!20!white,draw] (p4) {Lemma $3$}
      (5,-2) node[rectangle,rounded corners=7pt,fill=gray!20!white,draw] (p5) {Lemma $4$};
\end{scope}
\begin{scope}[yshift=2cm]
\path (5,-2) node[rectangle,rounded corners=7pt,fill=gray!20!white,draw] (p6) {Step $1$};
\end{scope}
\begin{scope}[yshift=-3.5cm, xshift=-0.2cm]
\path (1,0) node[rectangle,rounded corners=7pt,fill=gray!20!white,draw] (p7) {Step $2$}
      (2.45,0) node[rectangle,rounded corners=7pt,fill=gray!20!white,draw] (p8) {Step $3$};
\end{scope}
\begin{scope}[yshift=-3.5cm, xshift=6.75cm]
\path (1,0) node[rectangle,rounded corners=7pt,fill=gray!20!white,draw] (p9) {Step $4$}
      (2.45,0) node[rectangle,rounded corners=7pt,fill=gray!20!white,draw] (p10) {Step $5$};
\end{scope}
\begin{scope}[yshift=-5.5cm]
\path (2.5,0) node[rectangle,rounded corners=7pt,fill=gray!20!white,draw] (p11) {Step $6$};
\end{scope}
\begin{scope}[yshift=-5.5cm, xshift=5cm]
\path (2.5,0) node[rectangle,rounded corners=7pt,fill=gray!20!white,draw] (p12) {Step $6$};
\end{scope}
\begin{scope}[yshift=-8cm]
\path (0,0) node[rectangle,rounded corners=7pt,fill=gray!20!white,draw] (p13) {Step $7$};
\end{scope}
\begin{scope}[yshift=-8cm, xshift=7cm]
\path (3.0,0) node[rectangle,rounded corners=7pt,fill=gray!20!white,draw] (p14) {Step $8$};
\end{scope}
\begin{scope}[yshift=-10cm]
\path (-1,0) node[rectangle,rounded corners=7pt,fill=gray!20!white,draw] (p15) {Step $9$};
\end{scope}
\begin{scope}[yshift=-10cm, xshift=7cm]
\path (4,0) node[rectangle,rounded corners=7pt,fill=gray!20!white,draw] (p16) {Step $10$};
\end{scope}
\draw [postaction=decorate,decoration={markings, mark=at position 3.1cm with {\arrow[black]{stealth}}}, thick]  (p3.south) -- (p6.north);
\draw [postaction=decorate,decoration={markings, mark=at position 3.1cm with {\arrow[black]{stealth}}}, thick]  (p4.south) -- (p6.north);
\draw [postaction=decorate,decoration={markings, mark=at position 1.3cm with {\arrow[black]{stealth}}}, thick]  (p5.north) -- (p6.south);
\draw [postaction=decorate,decoration={markings, mark=at position 3.4cm with {\arrow[black]{stealth}}}, thick]  (p5.south) -- (p8.north west);
\draw [postaction=decorate,decoration={markings, mark=at position 2.7cm with {\arrow[black]{stealth}}}, thick]  (p1.south) -- (p7.north east);
\draw [postaction=decorate,decoration={markings, mark=at position 3.3cm with {\arrow[black]{stealth}}}, thick]  (p6.south west) -- (p7.north east);
\draw [postaction=decorate,decoration={markings, mark=at position 3.4cm with {\arrow[black]{stealth}}}, thick]  (p5.south) -- (p9.north east);
\draw [postaction=decorate,decoration={markings, mark=at position 2.7cm with {\arrow[black]{stealth}}}, thick]  (p2.south) -- (p10.north west);
\draw [postaction=decorate,decoration={markings, mark=at position 3.3cm with {\arrow[black]{stealth}}}, thick]  (p6.south east) -- (p10.north west);
\draw [postaction=decorate,decoration={markings, mark=at position 4.8cm with {\arrow[black]{stealth}}}, thick]  (p6.south west) -- (p11.north);
\draw [postaction=decorate,decoration={markings, mark=at position 1.5cm with {\arrow[black]{stealth}}}, thick]  (p8.south west) -- (p11.north);
\draw [postaction=decorate,decoration={markings, mark=at position 4.8cm with {\arrow[black]{stealth}}}, thick]  (p6.south east) -- (p12.north);
\draw [postaction=decorate,decoration={markings, mark=at position 1.5cm with {\arrow[black]{stealth}}}, thick]  (p9.south east) -- (p12.north);
\draw [postaction=decorate,decoration={markings, mark=at position 6.9cm with {\arrow[black]{stealth}}}, thick]  (p1.south) -- (p13.north);
\draw [postaction=decorate,decoration={markings, mark=at position 2.2cm with {\arrow[black]{stealth}}}, thick]  (p11.south) -- (p13.north);
\draw [postaction=decorate,decoration={markings, mark=at position 8cm with {\arrow[black]{stealth}}}, thick]  (p9.south east) -- (p13.north);
\draw [postaction=decorate,decoration={markings, mark=at position 6.9cm with {\arrow[black]{stealth}}}, thick]  (p2.south) -- (p14.north);
\draw [postaction=decorate,decoration={markings, mark=at position 2.2cm with {\arrow[black]{stealth}}}, thick]  (p12.south) -- (p14.north);
\draw [postaction=decorate,decoration={markings, mark=at position 8cm with {\arrow[black]{stealth}}}, thick]  (p8.south west) -- (p14.north);
\draw [postaction=decorate,decoration={markings, mark=at position 9cm with {\arrow[black]{stealth}}}, thick]  (p1.south) -- (p15.north);
\draw [postaction=decorate,decoration={markings, mark=at position 1cm with {\arrow[black]{stealth}}}, thick]  (p13.south) -- (p15.north);
\draw [postaction=decorate,decoration={markings, mark=at position 9cm with {\arrow[black]{stealth}}}, thick]  (p2.south) -- (p16.north);
\draw [postaction=decorate,decoration={markings, mark=at position 1cm with {\arrow[black]{stealth}}}, thick]  (p14.south) -- (p16.north);
\end{tikzpicture}}
\caption{Interrelation among the Steps $1$-$10$.}
\label{steps}
\end{figure}
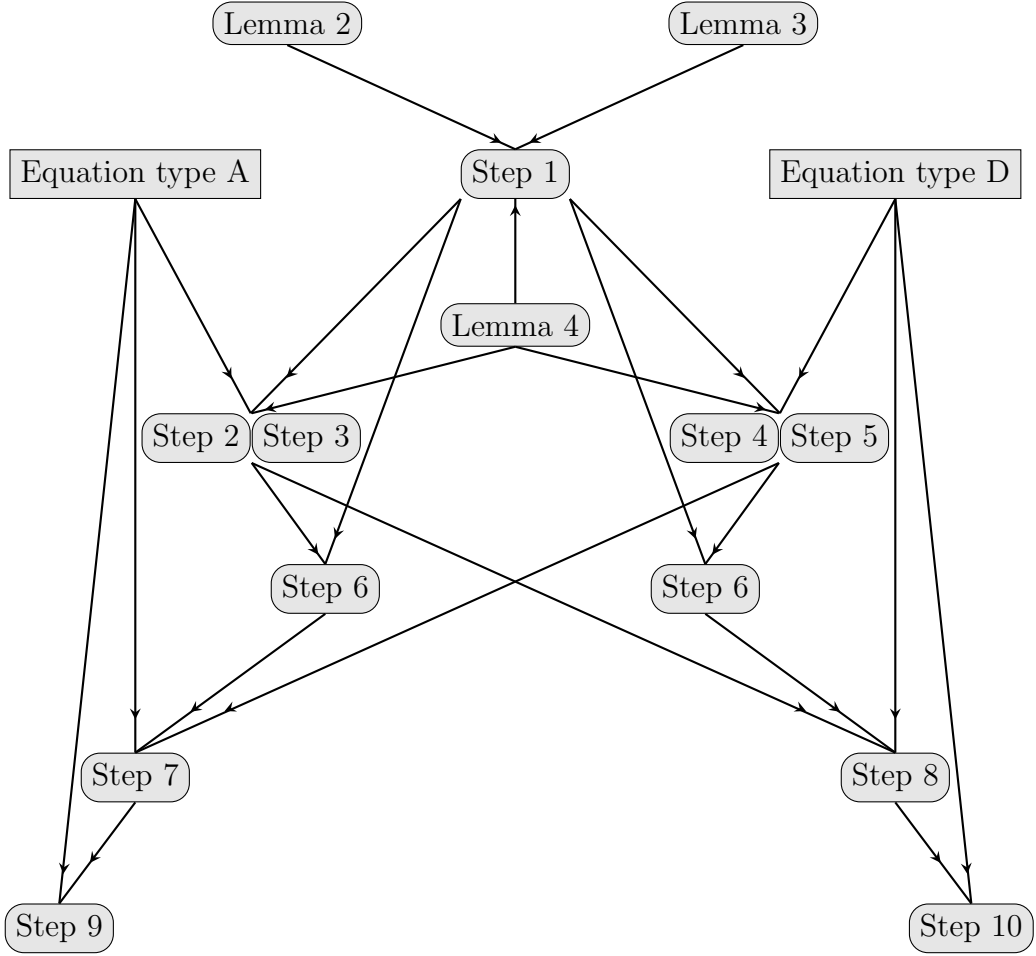

The sequence of steps $1$ through $10$ is interrelated and the dependence of a single step with the remaining ones
is schematically depicted in \Figref{steps}. Some comments are also required at this stage. For instance, in step $7$ we have
obtained formula (\ref{SXY}) expressing the scalar product $\mathcal{S}_n$ as a linear combination of auxiliary functions $V$. 
The coefficients of this linear combination is given by the function $\mathcal{K}_{\lambda \bar{\lambda}}$ defined in (\ref{KLbL}). Alternatively, in step $8$ we have also obtained
a similar formula expressing $\mathcal{S}_n$ in terms of the same functions $V$. This formula is given by (\ref{bSXY}) 
and it is expressed in terms of coefficients $\bar{\mathcal{K}}_{\lambda \bar{\lambda}}$ defined in (\ref{bKLbL}). 
Moreover, steps $9$ and $10$ tell us that the function $V$ satisfy the same system of equations (\ref{typeA}, \ref{typeD})
subjected to the mappings $L \mapsto L-1$, $n \mapsto n-1$ and $\mu_i \mapsto \mu_{i+1}$. This conclusion has been obtained
by inserting the expression (\ref{SXY}) back into the system (\ref{typeA}, \ref{typeD}). Alternatively, this same conclusion
can be obtained by replacing (\ref{bSXY}) into (\ref{typeA}, \ref{typeD}). In its turn, the function $V$ is essentially a
polynomial of order $2(L-1)$ as discussed in step $1$. Polynomial solutions of this type of linear functional equations are
unique as demonstrated in \cite{Galleas_2012} and this implies that $V$ is essentially the scalar product $\mathcal{S}_{n-1}$ on a lattice
of length $L-1$. Therefore, we can use (\ref{SXY}) or (\ref{bSXY}) to obtain $\mathcal{S}_n$ recursively starting
from the solution of (\ref{typeA}, \ref{typeD}) for $n =1$. In particular, the iteration procedure described by
(\ref{SXY}) and (\ref{bSXY}) seems to be naturally realized by multiple contour integrals. 

\vspace{0.3cm}
\noindent \textit{Multiple contour integrals.} In order to find an explicit expression realizing the iteration procedure described
by (\ref{SXY}), we assume the following ansatz for the scalar product $\mathcal{S}_n$,
\< \label{ansatz}
\mathcal{S}_n (\gen{X}^{1,n} | \gen{Y}^{1,n} ) = \oint \dots \oint   \frac{H(w_1, w_2, \dots, w_n | \bar{w}_1, \bar{w}_2, \dots, \bar{w}_n)}{\prod_{i,j=1}^n b(w_i - \lambda_j^{C}) b(\bar{w}_i - \lambda_j^{B})}  \prod_{k=1}^n \frac{\dd w_k}{2 \ii \pi} \frac{\dd \bar{w}_k}{2 \ii \pi} \; ,
\> 
where $H$ is a function to be determined. Also, the integration contour for each variable $w_i$ in the formula (\ref{ansatz})
must only enclose all variables in the set $\gen{X}^{1,n}$. Analogously, each variable $\bar{w}_i$ is integrated along a contour containing
only all variables in the set $\gen{Y}^{1,n}$. It is important to remark here that formula (\ref{ansatz}) concentrates the dependence with
all variables $\lambda_j^{\mathcal{B}}$ and $\lambda_j^{\mathcal{C}}$ in the denominator of its integrand. In this way, we are left with an
integral representation for $\mathcal{S}_n$ if we are able to exhibit a function $H$ implementing the iteration described by (\ref{SXY}).
This approach has been already discussed in \cite{Galleas_SCP} and here we shall restrict ourselves to presenting only the main 
points of this procedure. The substitution of (\ref{ansatz}) in (\ref{SXY}) then yields the following recursion relation for $H$,
\<
\label{HH}
&& H(w_1 , \dots , w_n | \bar{w}_1 , \dots , \bar{w}_n ) = \frac{b(2 w_1)}{a(2 w_1)} \frac{b(2 \bar{w}_1)}{a(2 \bar{w}_1)} R_1^{-1} \; \gen{det} (\gen{\Phi}^{(1)}) \bar{H} (w_2 , \dots , w_n | \bar{w}_2 , \dots , \bar{w}_n )  \nonumber \\
&& \left[ b(2 \mu_1 - 2\gamma) b(2 \mu_1 + \gamma) \prod_{\epsilon = \pm 1} b(h + \epsilon \mu_1) \prod_{j=2}^L a(\mu_1 - \epsilon \mu_j) a(\epsilon \mu_j -\mu_1) \right]^{-1} \nonumber \\
&& \times \prod_{k=2}^n b(w_k - w_1)^2 b(\bar{w}_k - \bar{w}_1)^2 a(w_k - \mu_1) b(w_k + \mu_1) a(\bar{w}_k - \mu_1) b(\bar{w}_k + \mu_1) \; .
\>
The function $\bar{H}$ in (\ref{HH}) consists of $H$, up to an overall multiplicative constant factor, under
the mappings $L \mapsto L-1$, $n \mapsto n-1$ and $\mu_i \mapsto \mu_{i+1}$. The function $R_1$ reads
\< 
\label{R1}
R_1 &\coloneqq& \prod_{k=1}^n \frac{a(w_k - \mu_1)}{b(w_k - \mu_1)} \frac{b(w_k + \mu_1)}{a(w_k + \mu_1)} - \prod_{k=1}^n \frac{a(\bar{w}_k - \mu_1)}{b(\bar{w}_k - \mu_1)} \frac{b(\bar{w}_k + \mu_1)}{a(\bar{w}_k + \mu_1)} \; ,
\>
while the $2 \times 2$ matrix $\gen{\Phi}^{(1)}$ has entries given by
\<
\label{Phi1}
\gen{\Phi}^{(1)}_{lm} &\coloneqq& \frac{b(h + s_1^m)}{\omega_l (s_1^m - (-1)^{l-1} \mu_1)} \prod_{k=1}^L a(s_1^m - \mu_k) a(s_1^m + \mu_k) \prod_{j=2}^n \frac{a(s_j^m - s_1^m)}{b(s_j^m - s_1^m)} \frac{b(s_j^m + s_1^m)}{a(s_j^m + s_1^m)} \nonumber \\
&& - \frac{a(s_1^m - h)}{\bar{\omega}_l (s_1^m + (-1)^{l-1} \mu_1)} \prod_{k=1}^L b(s_1^m - \mu_k) b(s_1^m + \mu_k) \prod_{j=2}^n \frac{a(s_1^m - s_j^m)}{b(s_1^m - s_j^m)} \frac{a(s_1^m + s_1^m + \gamma)}{b(s_1^m + s_j^m + \gamma)}  \nonumber \\
\> 
for variables $s_i^m \coloneqq \delta_{1m} w_i + \delta_{2m} \bar{w}_i$. The iteration of (\ref{HH}) starting with the solution
for the case $n =1$ obtained in the \Appref{sec:Ln1} yields formula (\ref{sol}). In fact, equations (\ref{typeA})
and (\ref{typeD}) are homogeneous and they can only determine the scalar product $\mathcal{S}_n$ up to an overall multiplicative
constant. This constant is then fixed by the asymptotic behavior (\ref{asymp}). This completes our proof.
\end{proof}

\begin{remark}
The relation (\ref{bSXY}) could also have been used instead of (\ref{SXY}) to produce a multiple contour
integral representation for the scalar product $\mathcal{S}_n$. However, one can notice that $\bar{\mathcal{K}}_{\lambda \bar{\lambda}}$ 
corresponds essentially to $\mathcal{K}_{\lambda \bar{\lambda}}$ under the map $\mu_i \mapsto - \mu_i$. This property is only
violated by the constant factors $\mathcal{C}_0$ and $\bar{\mathcal{C}}_0$. Since the overall constant factor needs to be fixed
by the asymptotic behavior (\ref{asymp}), this issue will not be relevant. In this way, the use of (\ref{bSXY}) produces
the same representation (\ref{sol}) under the map $\mu_i \mapsto - \mu_i$. This also implies that the transformation
$\mu_i \mapsto - \mu_i$ is a discrete symmetry of the scalar product $\mathcal{S}_n$.
\end{remark}

\section{Concluding remarks}
\label{sec:CONCLUSION}

This work was devoted to the analysis of scalar products of Bethe vectors associated
with the $XXZ$ spin chain with open boundary conditions. In particular, our study is based
on the description of scalar products by means of functional equations.

The desired scalar products are shown to satisfy a system of functional relations originated
from the reflection algebra. This approach for scalar products of Bethe vectors was firstly
proposed in \cite{Galleas_SCP} using the Yang-Baxter algebra as the source of functional equations.
The feasibility of using the reflection algebra in a similar way has been put forward recently
in \cite{Galleas_Lamers_2014}. Although the quantity computed in the present paper is different
from the one considered in \cite{Galleas_Lamers_2014}, the corresponding functional equations
are derived from the same reflection algebra relation. More precisely, the higher-degree relation
(\ref{AB}) is the same one employed in \cite{Galleas_Lamers_2014}. Here, however, we need to use a
total of four reflection algebra relations, namely (\ref{AB}), (\ref{CA}), (\ref{DB}) and (\ref{CD}), and
the main difference compared to \cite{Galleas_Lamers_2014} is the choice of map $\gen{\pi}_n$.
In the present work, two different realizations $\gen{\pi}_n^{\mathcal{B}, \mathcal{C}}$
was required in order to describe the desired scalar product.  

The methodology used to solve this type of functional equations has been introduced in \cite{Galleas_2011}
for an equation describing the partition function of a Solid-on-Solid model with domain wall boundaries.
The key point of this method is that the location of special zeroes of the partition function or scalar products
induces a separation of variables. This separation of variables can also be regarded as a recurrence relation allowing
us to readily obtain the desired solution. Here the solution of our system of equations is obtained in a compact form.
It is given by a multiple contour integral as stated in Theorem \ref{off-shell}. Interestingly, this type of integral formula seems to be
a general structure accommodating quantities such as scalar products and partition functions associated with
integrable vertex models. For instance, similar integral formulae have appeared in \cite{deGier_Galleas_2011, Galleas_2012, Galleas_2013, Galleas_Lamers_2014}
for partition functions with domain-wall boundaries and in \cite{deGier_Galleas_2011, Galleas_SCP} for scalar products.

Functional equations such as (\ref{typeA}) and (\ref{typeD}) also to encode a family of partial differential equations
describing certain quantities of interest. This has been shown in \cite{Galleas_2011, Galleas_proc, Galleas_2014, Galleas_Lamers_2014}
by recasting our functional relations in an operatorial form. We have not pursued that direction in the present paper but we hope
to report on that in a future publication.

\section{Acknowledgements}
\label{sec:ACK}
The author thanks G. Arutyunov for discussions and comments on this manuscript.
The work of W.G. was supported by the German Science Foundation (DFG) under the Collaborative
Research Center (SFB) 676 Particles, Strings and the Early Universe.

\appendix

\section{$\mathcal{S}_n$ as a doubly symmetric function}
\label{sec:SYM}

This appendix is concerned with the proof of Lemma \ref{symmetry}. This lemma states that the function 
$\mathcal{S}_n ( \lambda_1^{\mathcal{C}}, \dots , \lambda_n^{\mathcal{C}} | \lambda_1^{\mathcal{B}}, \dots , \lambda_n^{\mathcal{B}})$
is invariant under the permutation of variables $\lambda_i^{\mathcal{C}} \leftrightarrow \lambda_j^{\mathcal{C}}$
and $\lambda_k^{\mathcal{B}} \leftrightarrow \lambda_l^{\mathcal{B}}$ independently. We thus say that $\mathcal{S}_n$ is a doubly symmetric function. 
The definition of $\mathcal{S}_n$, as given in (\ref{scp}), and the commutation
rules $\left[ \mathcal{B}(\lambda_1) , \mathcal{B}(\lambda_2) \right] = \left[ \mathcal{C}(\lambda_1) , \mathcal{C}(\lambda_2) \right] = 0$
encoded in the reflection algebra, already requires this property to be fulfilled. However, once we assume the scalar product
to satisfy equations (\ref{typeA}) and (\ref{typeD}), we would like their solution to naturally inherit this symmetry property.
That is precisely what we intend to show here and our proof will follow the same arguments employed in \cite{Galleas_SCP, Galleas_Lamers_2014}. 

Firstly, we assume $\mathcal{S}_n$ to be an analytic function and inspect (\ref{typeA}) in the limit 
$\lambda_0 \to \lambda_k^{\mathcal{C}}$. From (\ref{coeffA}) we can see that $M_0$ and $N_{\lambda_k^{\mathcal{C}}}^{(\mathcal{C})}$
are the only singular coefficients in that limit. Thus the integration of (\ref{typeA}) with respect to the variable
$\lambda_0$ around a contour enclosing solely the variable $\lambda_k^{\mathcal{C}}$ yields the following
identity,
\<
&& \gen{Res}_{\lambda_0 = \lambda_k^{\mathcal{C}}} \left( M_0 \right) \; \mathcal{S}_n (\lambda_1^{\mathcal{C}} , \dots , \lambda_{k-1}^{\mathcal{C}}, \lambda_{k}^{\mathcal{C}}, \lambda_{k+1}^{\mathcal{C}}, \dots , \lambda_{n}^{\mathcal{C}} \; | \; \vec{Y}^{1,n} ) \nonumber \\
&& \qquad \qquad + \; \gen{Res}_{\lambda_0 = \lambda_k^{\mathcal{C}}} \left( N_{\lambda_k^{\mathcal{C}}}^{(\mathcal{C})} \right) \; \mathcal{S}_n (\lambda_{k}^{\mathcal{C}} , \lambda_1^{\mathcal{C}} , \dots , \lambda_{k-1}^{\mathcal{C}}, \lambda_{k+1}^{\mathcal{C}}, \dots , \lambda_{n}^{\mathcal{C}} \; | \; \vec{Y}^{1,n} ) = 0 \; . \nonumber \\
\>
Moreover, from formulae (\ref{coeffA}) we can readily verify the property $\gen{Res}_{\lambda_0 = \lambda_k^{\mathcal{C}}} \left( M_0 \right) = - \gen{Res}_{\lambda_0 = \lambda_k^{\mathcal{C}}} \left( N_{\lambda_k^{\mathcal{C}}}^{(\mathcal{C})} \right)$
which allows us to conclude that 
\[
\label{cyclicC}
\mathcal{S}_n (\lambda_1^{\mathcal{C}} , \dots , \lambda_{k-1}^{\mathcal{C}}, \lambda_{k}^{\mathcal{C}}, \lambda_{k+1}^{\mathcal{C}}, \dots , \lambda_{n}^{\mathcal{C}} \; | \; \vec{Y}^{1,n} )  = \mathcal{S}_n (\lambda_{k}^{\mathcal{C}} , \lambda_1^{\mathcal{C}} , \dots , \lambda_{k-1}^{\mathcal{C}}, \lambda_{k+1}^{\mathcal{C}}, \dots , \lambda_{n}^{\mathcal{C}} \; | \; \vec{Y}^{1,n} ) \; .
\]
Next we integrate Eq. (\ref{typeA}) over the variable $\lambda_0$ around a contour containing
solely the variable $\lambda_k^{\mathcal{B}}$. Similarly to (\ref{cyclicC}), this procedure allows us to conclude that
\[
\label{cyclicB}
\mathcal{S}_n (\vec{X}^{1,n} \; | \; \lambda_1^{\mathcal{B}} , \dots , \lambda_{k-1}^{\mathcal{B}}, \lambda_{k}^{\mathcal{B}}, \lambda_{k+1}^{\mathcal{B}}, \dots , \lambda_{n}^{\mathcal{B}}  ) = \mathcal{S}_n (\vec{X}^{1,n} \; | \;   \lambda_{k}^{\mathcal{B}} , \lambda_1^{\mathcal{B}} , \dots , \lambda_{k-1}^{\mathcal{B}}, \lambda_{k+1}^{\mathcal{B}}, \dots , \lambda_{n}^{\mathcal{B}}  ) \; .
\]
The relations (\ref{cyclicC}) and (\ref{cyclicB}) tell us that $\mathcal{S}_n$ is invariant under cyclic permutations of
$\lambda_1^{\mathcal{C},\mathcal{B}} , \dots , \lambda_k^{\mathcal{C},\mathcal{B}}$ for any $k$ in the interval $1 \leq k \leq n$. 
Then, since the symmetric group of order $n$ is generated by any cycle of length $n$, in addition to any single transposition, we can conclude
that $\mathcal{S}_n$ is invariant under the action of $\gen{Sym} \left( \gen{X}^{1,n} \right) \otimes \gen{Sym} \left( \gen{Y}^{1,n} \right)$. 
This concludes our proof. Although here we have considered only Eq. (\ref{typeA}), it is worth remarking that the same results
could have been obtained from a similar analysis of Eq. (\ref{typeD}).

\section{Polynomial structure}
\label{sec:POL}

In order to prove Lemma \ref{polynomial} it is convenient to introduce the following
extra definitions. 

\begin{mydef}
Let $\mathbb{C}[x]$ denote the polynomial ring in the variable $x$. Then recall the definitions
$y_i \coloneqq e^{2 \mu_i}$, $q \coloneqq e^{\gamma}$ and $t \coloneqq e^{h}$, and let
$\mathbb{C}[y_1^{\pm 1} , y_2^{\pm 1} , \dots , y_n^{\pm 1}, q^{\pm 1} , t^{\pm 1}]$ denote the space
of regular functions in the variables $y_1, \dots , y_n, q, t$. The space of polynomials in the variable
$x$ with coefficients in $\mathbb{C}[y_1^{\pm 1} , y_2^{\pm 1} , \dots , y_n^{\pm 1}, q^{\pm 1} , t^{\pm 1}]$
is then denoted by $\mathbb{C}[y_1^{\pm 1} , y_2^{\pm 1} , \dots , y_n^{\pm 1}, q^{\pm 1} , t^{\pm 1}][x]$.
Lastly we define $\mathbb{P}_m [x] \subseteq \mathbb{C}[y_1^{\pm 1} , y_2^{\pm 1} , \dots , y_n^{\pm 1}, q^{\pm 1} , t^{\pm 1}][x]$  
as the subspace formed by polynomials of degree $m$ in the variable $x$.
\end{mydef}

As we can see from (\ref{scp}), the dependence of $\mathcal{S}_n$ with the set of variables $\gen{X}^{1,n}$ and $\gen{Y}^{1,n}$
is only due to the operators $\mathcal{C}$ and $\mathcal{B}$ respectively. More precisely, the whole dependence with the 
variable $\lambda_i^{\mathcal{C}}$ comes from the operator $\mathcal{C}(\lambda_i^{\mathcal{C}})$ entering the definition (\ref{scp}), while
the dependence with $\lambda_i^{\mathcal{B}}$ is characterized by the operator $\mathcal{B}(\lambda_i^{\mathcal{B}})$. 
Therefore, Lemma \ref{polynomial} follows from the functional dependence of these operators with their variables. Using the notation 
$x \coloneqq e^{2 \lambda}$, it is then sufficient to show that
\<
\label{BC}
\mathcal{B}(x) = x^{-L} f_{\mathcal{B}}^{2 L}(x) \qquad \mbox{and} \qquad \mathcal{C}(x) = x^{-L} f_{\mathcal{C}}^{2 L}(x) \; ,
\>
where $f_{\mathcal{B}, \mathcal{C}}^{2 L}(x) \in \mathbb{P}_{2L}[x] \otimes \gen{End}\left( (\mathbb{C}^2)^{\otimes L} \right)$. 
The proof of (\ref{BC}) for the operator $\mathcal{B}$ can be found in \cite{Galleas_Lamers_2014} and we shall not repeat it here.
On the other hand, the functional dependence (\ref{BC}) for the operator $\mathcal{C}$ is a direct consequence of the proof presented in 
\cite{Galleas_Lamers_2014} and the fact that $\mathcal{C}(\lambda)$ corresponds to $\mathcal{B}(\lambda)^t$ under the map $\mu_i \to -\mu_i$ as demonstrated in 
\cite{Galleas_2008}. In what follows we shall describe the main points leading to the desired property.

\begin{mydef}
Let $A, B, C, D, \bar{A}, \bar{B}, \bar{C}, \bar{D} \colon \mathbb{C} \to \gen{End} \left( (\mathbb{C}^2)^{\otimes L} \right)$
be operator-valued functions defined as
\<
\label{mono}
\mathop{\overleftarrow\prod}\limits_{1 \leq j \leq L} \mathcal{R}_{0 j}(\lambda - \mu_j) &\eqqcolon& 
\left(  \begin{matrix}
A(\lambda) & B(\lambda) \cr
C(\lambda) & D(\lambda) \end{matrix} \right) \nonumber \\
\mathop{\overrightarrow\prod}\limits_{1 \leq j \leq L} \mathcal{R}_{0 j}(\lambda - \mu_j) &\eqqcolon& 
\left(  \begin{matrix}
\bar{A}(\lambda) & \bar{B}(\lambda) \cr
\bar{C}(\lambda) & \bar{D}(\lambda) \end{matrix} \right) \; .
\>
\end{mydef}
Taking into account (\ref{ABCD}) and (\ref{mono}) we can then write
\<
\label{BBCC}
\mathcal{B}(\lambda) &=& \sinh{(h + \lambda)} A(\lambda) \bar{B}(\lambda) + \sinh{(h - \lambda)} B(\lambda) \bar{D}(\lambda) \nonumber \\
\mathcal{C}(\lambda) &=& \sinh{(h + \lambda)} C(\lambda) \bar{A}(\lambda) + \sinh{(h - \lambda)} D(\lambda) \bar{C}(\lambda) \; .
\>
Moreover, in \cite{Galleas_2008} we have demonstrated that the operators defined in (\ref{mono}) satisfy the following properties:
\begin{align}
\label{Tmono}
A(\lambda)^t =& \left. \bar{A}(\lambda) \right|_{\mu_i \to -\mu_i} & B(\lambda)^t &= \left. \bar{C}(\lambda) \right|_{\mu_i \to -\mu_i} \cr
C(\lambda)^t =& \left. \bar{B}(\lambda) \right|_{\mu_i \to -\mu_i} & D(\lambda)^t &= \left. \bar{D}(\lambda) \right|_{\mu_i \to -\mu_i} \; ,
\end{align}
as a consequence of the crossing symmetry exhibited by the $\mathcal{U}_q [\widehat{\alg{sl}}(2)]$ 
invariant $\mathcal{R}$-matrix employed in the present work. Using (\ref{BBCC}) and (\ref{Tmono}) we can
readily show that the relation $\mathcal{C}(\lambda) = \left. \mathcal{B}(\lambda)^t \right|_{\mu_i \to - \mu_i}$ holds,
which concludes our proof.

\section{Asymptotic behavior}
\label{sec:ASYMP}

The full determination of the scalar product $\mathcal{S}_n$ as solution of the system of functional
equations (\ref{typeA}, \ref{typeD}) requires we are able to evaluate $\mathcal{S}_n$ for a particular value
of its variables. This is due to the fact that (\ref{typeA}) and (\ref{typeD}) are homogeneous equations and, therefore,
they are only able to determine the solution up to an overall multiplicative factor independent of the variables
$\lambda_i^{\mathcal{B} , \mathcal{C} }$. Here we find that the points $x_i^{\mathcal{B} , \mathcal{C} } \coloneqq e^{2 \lambda_i^{\mathcal{B} , \mathcal{C} }} \to \infty$
are suitable for that purpose, and the first step to compute $\mathcal{S}_n$ in that limit is the analysis of the behavior
of $\mathcal{B}(x)$ and $\mathcal{C}(x)$ for $x \to \infty$.

In \cite{Galleas_Lamers_2014} we have shown that the operator $\mathcal{B}(x)$ exhibits the asymptotic behavior
\[ \label{BX}
\mathcal{B}(x) \sim \frac{q^{L-1}}{2^{2L+1}} (q - q^{-1}) x^L \sum_{j=1}^L \left( P_j^{+} + P_j^{-} \right) \qquad \mbox{for} \; x \to \infty \; ,
\]
where 
\[
\label{PJ}
P_j^{\pm} \coloneqq \pm ( t y_j^{\frac{1}{2}} )^{\pm 1} \mbox{id}^{\otimes(j-1)} \otimes X^{-} \otimes \left( K^{\pm 1} \right)^{\otimes (L-j)} \; .
\]
Here we are using the conventions $q \coloneqq e^{\gamma}$, $t \coloneqq e^{h}$, $y_i \coloneqq e^{2 \mu_i}$
and, due to the property $\mathcal{C}(\lambda) = \left. \mathcal{B}(\lambda)^t \right|_{\mu_i \to - \mu_i}$
shown in \Appref{sec:POL}, we readily find that
\[ \label{CX}
\mathcal{C}(x) \sim \frac{q^{L-1}}{2^{2L+1}} (q - q^{-1}) x^L \sum_{j=1}^L \left( \bar{P}_j^{+} + \bar{P}_j^{-} \right) 
\]
in the limit $x \to \infty$. The operators $\bar{P}_j^{\pm}$ appearing in (\ref{CX}) are in their turn defined as
\[
\label{bPJ}
\bar{P}_j^{\pm} \coloneqq \pm ( t y_j^{-\frac{1}{2}} )^{\pm 1} \mbox{id}^{\otimes(j-1)} \otimes X^{+} \otimes \left( K^{\pm 1} \right)^{\otimes (L-j)} \; .
\]
The symbol $\mbox{id}$ in (\ref{PJ}) and (\ref{bPJ}) stands for the identity matrix in $\gen{End}(\mathbb{C}^2)$,
while $X^{\pm}$ and $K$ denote the generators of the $\mathcal{U}_q[ \alg{sl}(2) ]$ algebra in the fundamental
representation as described in \cite{Galleas_Lamers_2014}.
As a consequence of the $\mathcal{U}_q[ \alg{sl}(2) ]$ algebra satisfied by $X^{\pm}$ and $K$, one can show that the
operators $P_j^{\pm}$ are subjected to the following commutation rules:
\begin{align}
\label{Ppm}
P_i^{\pm} P_j^{\pm} &= q^{\mp2}  P_j^{\pm} P_i^{\pm} \; , \qquad P_i^{\pm} P_j^{\mp} = q^{\mp 2}  P_j^{\mp} P_i^{\pm} \; , && \text{for } \; i<j \; , \nonumber \\
P_i^{s} P_i^{s'} & =0  \; && \text{for } \;  s,s' \in \{\pm\} \; . 
\end{align}
The operators $\bar{P}_j^{\pm}$ obey similar relations, namely
\begin{align}
\label{bPpm}
\bar{P}_i^{\pm} \bar{P}_j^{\pm} &= q^{\pm2}  \bar{P}_j^{\pm} \bar{P}_i^{\pm} \; , \qquad \bar{P}_i^{\pm} \bar{P}_j^{\mp} = q^{\pm 2}  \bar{P}_j^{\mp} \bar{P}_i^{\pm} \; , && \text{for } \; i<j \; , \nonumber \\
\bar{P}_i^{s} \bar{P}_i^{s'} & =0  \; && \text{for } \;  s,s' \in \{\pm\} \; . 
\end{align}

For latter convenience we also define the operators
\<
\label{QQ}
Q_j^{(m)} \coloneqq P_j^{+} q^{-2m} + P_j^{-} q^{2m} \qquad  \qquad \bar{Q}_j^{(m)} \coloneqq \bar{P}_j^{+} q^{2m} + \bar{P}_j^{-} q^{-2m} \; ,
\>
in such a way that $Q_j^{(m)}$ and $\bar{Q}_j^{(m)}$ satisfy the following set of commutation
relations,
\begin{align}
\label{QbQ}
Q_i^{(m)} Q_j^{(n)} &= Q_j^{(n)} Q_i^{(m+1)}  & \bar{Q}_i^{(m)} \bar{Q}_j^{(n)} &= \bar{Q}_j^{(n)} \bar{Q}_i^{(m+1)}  \qquad  \text{for } \; i<j  \nonumber \\
Q_i^{(m)} Q_i^{(n)} &= 0 &  \bar{Q}_i^{(m)} \bar{Q}_i^{(n)} &= 0 \; ,
\end{align}
as a consequence of (\ref{Ppm}) and (\ref{bPpm}), 

Now the expressions (\ref{BX}), (\ref{CX}) and (\ref{QQ}) allow us to write
\<
\label{Pasymp}
\bar{\mathcal{S}}_n &\sim& \frac{q^{2n (L-1)}}{2^{2n(2L+1)}} (q - q^{-1})^{2n} \bra{0} \bar{\mathfrak{J}} \mathfrak{J} \ket{0} \prod_{i=1}^{n} \left( x_i^{\mathcal{B}} x_i^{\mathcal{C}} \right)^{2L} 
\>
in the limit $x_i^{\mathcal{B} , \mathcal{C} } \to \infty$ for all $i \in \{1, 2, \dots , n\}$. Formula (\ref{Pasymp}) is given in terms of the operators $\mathfrak{J}$ and 
$\bar{\mathfrak{J}}$ defined respectively as
\<
\mathfrak{J} \coloneqq \sum_{r_1 = 1}^L \dots \sum_{r_n = 1}^L \mathop{\overrightarrow\prod}\limits_{1 \le s \le n } Q_{r_s}^{(0)} 
\>
and 
\<
\bar{\mathfrak{J}} \coloneqq \sum_{\bar{r}_1 = 1}^L \dots \sum_{\bar{r}_n = 1}^L \mathop{\overrightarrow\prod}\limits_{1 \le s \le n } \bar{Q}_{\bar{r}_s}^{(0)} \; .
\>
Next we use the methodology described in \cite{Galleas_Lamers_2014} to find that $\mathfrak{J}$ and $\bar{\mathfrak{J}}$
can be rewritten as
\<
\label{JbJ}
\mathfrak{J} = \sum_{1 \leq r_i < r_{i+1}  \leq L} \mathop{\overleftarrow\prod}\limits_{1 \le s \le n } \left( \sum_{l=0}^{n-s} Q_{r_s}^{(l)} \right)
\qquad \mbox{and} \qquad \bar{\mathfrak{J}} = \sum_{1 \leq \bar{r}_i < \bar{r}_{i+1} \leq L} \mathop{\overleftarrow\prod}\limits_{1 \le s \le n } \left( \sum_{l=0}^{n-s} \bar{Q}_{\bar{r}_s}^{(l)} \right) \; . \nonumber \\
\>
The above expressions are more manageable and the definitions (\ref{QQ}) allow us to write the summations in (\ref{JbJ}) as
\<
\label{QbQ}
\sum_{l=0}^{n-s} Q_{r_s}^{(l)} &=& P_{r_s}^{+} \Delta_{n-s}^{-} + P_{r_s}^{-} \Delta_{n-s}^{+} \nonumber \\
\sum_{l=0}^{n-s} \bar{Q}_{\bar{r}_s}^{(l)} &=& \bar{P}_{\bar{r}_s}^{+} \Delta_{n-s}^{+} + \bar{P}_{\bar{r}_s}^{-} \Delta_{n-s}^{-} \; ,
\>
where $\Delta_l^{\pm}  \coloneqq  \sum_{m=0}^{l} q^{\pm 2m}$.

In order to find an explicit expression for the leading coefficient (\ref{Pasymp}) we still need to evaluate 
the quantity $\bra{0} \bar{\mathfrak{J}} \mathfrak{J} \ket{0}$. This task can be readily performed with the help
of (\ref{PJ}), (\ref{bPJ}), (\ref{JbJ}) and (\ref{QbQ}). By doing so we find 
\<
\label{OJJO}
\bra{0} \bar{\mathfrak{J}} \mathfrak{J} \ket{0} =  \sum_{1 \leq r_i < r_{i+1}  \leq L} \prod_{s=1}^{n} \prod_{\epsilon \in \{ \pm \}} \left( t y_{r_s}^{- \epsilon \frac{1}{2}} q^{L - r_s} \Delta_{n-s}^{\epsilon} - t^{-1} y_{r_s}^{\epsilon \frac{1}{2}} q^{r_s - L} \Delta_{n-s}^{-\epsilon} \right) \; , \nonumber \\
\>
which leads us directly to the statement of Lemma \ref{asymptotic}.

\section{Special zeroes}
\label{sec:ZEROES}

The resolution of the system of Eqs. (\ref{typeA}, \ref{typeD}) follows a sequence of systematic
steps relying heavily on Lemma \ref{zeroes}. This lemma gives us the location of special zeroes
of the scalar product $\mathcal{S}_n$ inducing the separation of variables expressed by the 
relations (\ref{SXY}) and (\ref{bSXY}). More precisely, we have considered three pairs of
points, namely $( \lambda_1^{\mathcal{C}} , \lambda_2^{\mathcal{C}} ) = ( \mu_1 - \gamma , \mu_1 )$,
$( \lambda_1^{\mathcal{C}} , \lambda_2^{\mathcal{C}} ) = ( \mu_1 - \gamma , -\mu_1 - \gamma )$
and $( \lambda_1^{\mathcal{C}} , \lambda_2^{\mathcal{C}} ) = ( -\mu_1 , \mu_1 )$,  for which the specialization
of $\mathcal{S}_n$ vanishes. The same property with variables $\lambda_i^{\mathcal{C}}$ exchanged by 
$\lambda_i^{\mathcal{B}}$ also holds.

The location of these special points is encoded in our system of equations and this is what we intend to unveil here.
Also, we shall only present a detailed proof of this property for $( \lambda_1^{\mathcal{C}} , \lambda_2^{\mathcal{C}} ) = ( \mu_1 - \gamma , \mu_1 )$
as the remaining cases can be obtained along the same lines. For illustrative purposes we shall firstly discuss the 
cases $n =2,3$, and only then describe the general case. 

\paragraph{Case $n=2$.} We consider (\ref{typeA}) under the specializations $\lambda_1^{\mathcal{C}} = \mu_1 - \gamma$
and $\lambda_2^{\mathcal{C}} = \mu_1$. Under this particular specialization one can verify that the coefficients
$N_{\lambda_1^{\mathcal{C}}}^{(\mathcal{C})}$ and $N_{\lambda_2^{\mathcal{C}}}^{(\mathcal{C})}$ vanish. We are then left with
the equation
\<
\label{A2}
\left. M_0 \right|_{1,2} \mathcal{S}_2 (\mu_1 - \gamma, \mu_1 | \lambda_1^{\mathcal{B}} , \lambda_2^{\mathcal{B}} ) &+& \left. N_{\lambda_1^{\mathcal{B}}}^{(\mathcal{B})} \right|_{1,2} \mathcal{S}_2 (\mu_1 - \gamma, \mu_1 | \lambda_0 , \lambda_2^{\mathcal{B}} ) \nonumber \\
&+& \left. N_{\lambda_2^{\mathcal{B}}}^{(\mathcal{B})} \right|_{1,2} \mathcal{S}_2 (\mu_1 - \gamma, \mu_1 | \lambda_0 , \lambda_1^{\mathcal{B}} ) = 0 \; ,
\>
where $\left. \right|_{1,2}$ denotes the prescribed specialization of $\lambda_1^{\mathcal{C}}$ and
$\lambda_2^{\mathcal{C}}$. Next we consider Remark \ref{rema} and apply the maps $\lambda_0 \leftrightarrow \lambda_1^{\mathcal{B}}$
and $\lambda_0 \leftrightarrow \lambda_2^{\mathcal{B}}$ taking into account Lemma \ref{symmetry}. This produces two additional equations with
the same structure of (\ref{A2}) but modified coefficients. We then define the coefficients
\[
N_j^{i} \coloneqq \begin{cases}
\left( \left. M_0 \right|_{1,2} \right)_{\lambda_0 \leftrightarrow \lambda_i^{\mathcal{B}} } \qquad \quad \mbox{for} \; j=0 \cr
\left( \left. N_{\lambda_j^{\mathcal{B}}}^{(\mathcal{B})} \right|_{1,2}  \right)_{\lambda_0 \leftrightarrow \lambda_i^{\mathcal{B}} } \qquad \mbox{for} \; j=1,2 \cr
\end{cases}
\]
in such a way that the resulting equations can be conveniently written as
\<
\left( \begin{matrix}
N_0^{(0)} & N_1^{(0)} & N_2^{(0)} \cr
N_1^{(1)} & N_0^{(1)} & N_2^{(1)} \cr
N_2^{(2)} & N_1^{(2)} & N_0^{(2)} \end{matrix} \right)
\left( \begin{matrix}
\mathcal{S}_2 (\mu_1 - \gamma, \mu_1 | \lambda_1^{\mathcal{B}} , \lambda_2^{\mathcal{B}} ) \cr
\mathcal{S}_2 (\mu_1 - \gamma, \mu_1 | \lambda_0 , \lambda_2^{\mathcal{B}} ) \cr
\mathcal{S}_2 (\mu_1 - \gamma, \mu_1 | \lambda_0 , \lambda_1^{\mathcal{B}} )
\end{matrix} \right) = 0 \; .
\>
Now one can verify that $\mbox{det}\left( N_j^i \right) \neq 0$ for generic values of the variables and this allows us to conclude that 
$\mathcal{S}_n (\mu_1 - \gamma, \mu_1 | \lambda_1^{\mathcal{B}} , \lambda_2^{\mathcal{B}} )$ vanishes.
It is worth remarking here that we have used only Eq. (\ref{typeA}) for the case $n=2$. For $n > 2$ we shall also need
(\ref{typeD}).

\paragraph{Case $n=3$.} In that case we look at (\ref{typeA}) and (\ref{typeD}) under the specializations
$\lambda_2^{\mathcal{C}} = \mu_1 - \gamma$ and $\lambda_3^{\mathcal{C}} = \mu_1$. The coefficients 
$N_{\lambda_2^{\mathcal{C}}}^{(\mathcal{C})}$, $N_{\lambda_3^{\mathcal{C}}}^{(\mathcal{C})}$,
$\widetilde{N}_{\lambda_2^{\mathcal{C}}}^{(\mathcal{C})}$ and $\widetilde{N}_{\lambda_3^{\mathcal{C}}}^{(\mathcal{C})}$
vanish for this particular specialization and we are left with the following equations:
\< \label{A3}
&& \left. M_0 \right|_{2,3} \mathcal{S}_3 ( \lambda_1^{\mathcal{C}} , \mu_1 - \gamma, \mu_1 | \lambda_1^{\mathcal{B}}, \lambda_2^{\mathcal{B}} , \lambda_3^{\mathcal{B}}  )
+ \left. N_1^{(\mathcal{C})} \right|_{2,3} \mathcal{S}_3 ( \lambda_0 , \mu_1 - \gamma, \mu_1 | \lambda_1^{\mathcal{B}}, \lambda_2^{\mathcal{B}} , \lambda_3^{\mathcal{B}}  ) \nonumber \\
&&+ \left. N_1^{(\mathcal{B})} \right|_{2,3} \mathcal{S}_3 ( \lambda_1^{\mathcal{C}} , \mu_1 - \gamma, \mu_1 | \lambda_0, \lambda_2^{\mathcal{B}} , \lambda_3^{\mathcal{B}}  )
+ \left. N_2^{(\mathcal{B})} \right|_{2,3} \mathcal{S}_3 ( \lambda_1^{\mathcal{C}} , \mu_1 - \gamma, \mu_1 | \lambda_0, \lambda_1^{\mathcal{B}} , \lambda_3^{\mathcal{B}}  ) \nonumber \\
&& + \left. N_3^{(\mathcal{B})} \right|_{2,3} \mathcal{S}_3 ( \lambda_1^{\mathcal{C}} , \mu_1 - \gamma, \mu_1 | \lambda_0, \lambda_1^{\mathcal{B}} , \lambda_2^{\mathcal{B}}  ) = 0
\>
\< \label{D3}
&& \left. \widetilde{M}_0 \right|_{2,3} \mathcal{S}_3 ( \lambda_1^{\mathcal{C}} , \mu_1 - \gamma, \mu_1 | \lambda_1^{\mathcal{B}}, \lambda_2^{\mathcal{B}} , \lambda_3^{\mathcal{B}}  )
+ \left. \widetilde{N}_1^{(\mathcal{C})} \right|_{2,3} \mathcal{S}_3 ( \lambda_0 , \mu_1 - \gamma, \mu_1 | \lambda_1^{\mathcal{B}}, \lambda_2^{\mathcal{B}} , \lambda_3^{\mathcal{B}}  ) \nonumber \\
&&+ \left. \widetilde{N}_1^{(\mathcal{B})} \right|_{2,3} \mathcal{S}_3 ( \lambda_1^{\mathcal{C}} , \mu_1 - \gamma, \mu_1 | \lambda_0, \lambda_2^{\mathcal{B}} , \lambda_3^{\mathcal{B}}  )
+ \left. \widetilde{N}_2^{(\mathcal{B})} \right|_{2,3} \mathcal{S}_3 ( \lambda_1^{\mathcal{C}} , \mu_1 - \gamma, \mu_1 | \lambda_0, \lambda_1^{\mathcal{B}} , \lambda_3^{\mathcal{B}}  ) \nonumber \\
&& + \left. \widetilde{N}_3^{(\mathcal{B})} \right|_{2,3} \mathcal{S}_3 ( \lambda_1^{\mathcal{C}} , \mu_1 - \gamma, \mu_1 | \lambda_0, \lambda_1^{\mathcal{B}} , \lambda_2^{\mathcal{B}}  ) = 0
\>
Similarly to the case $n=2$, in (\ref{A3}) and (\ref{D3}) we have used the symbol $\left. \right|_{2,3}$ to denote the 
aforementioned specializations of $\lambda_2^{\mathcal{C}}$ and $\lambda_3^{\mathcal{C}}$. We then eliminate the term
$\mathcal{S}_3 ( \lambda_0 , \mu_1 - \gamma, \mu_1 | \lambda_1^{\mathcal{B}}, \lambda_2^{\mathcal{B}} , \lambda_3^{\mathcal{B}}  ) $
from the system of equations formed by (\ref{A3}) and (\ref{D3}). In addition to that we also consider the maps 
$\lambda_0 \leftrightarrow \lambda_i^{\mathcal{B}}$ for $1 \leq i \leq 3$ to produce three extra equations. We are thus left
with a total of four linear equations which can be conveniently written as
\<
\label{CM4}
\left( \begin{matrix}
N_0^{0} & N_1^{0} & N_2^{0} & N_3^{0} \cr
N_1^{1} & N_0^{1} & N_2^{1} & N_3^{1} \cr
N_2^{2} & N_1^{2} & N_0^{2} & N_3^{2} \cr
N_3^{3} & N_1^{3} & N_2^{3} & N_0^{3} \end{matrix} \right)
\left( \begin{matrix}
\mathcal{S}_3 ( \lambda_1^{\mathcal{C}} , \mu_1 - \gamma, \mu_1 | \lambda_1^{\mathcal{B}}, \lambda_2^{\mathcal{B}} , \lambda_3^{\mathcal{B}}  ) \cr
\mathcal{S}_3 ( \lambda_1^{\mathcal{C}} , \mu_1 - \gamma, \mu_1 | \lambda_0, \lambda_2^{\mathcal{B}} , \lambda_3^{\mathcal{B}}  ) \cr
\mathcal{S}_3 ( \lambda_1^{\mathcal{C}} , \mu_1 - \gamma, \mu_1 | \lambda_0, \lambda_1^{\mathcal{B}} , \lambda_3^{\mathcal{B}}  ) \cr
\mathcal{S}_3 ( \lambda_1^{\mathcal{C}} , \mu_1 - \gamma, \mu_1 | \lambda_0, \lambda_1^{\mathcal{B}} , \lambda_2^{\mathcal{B}}  ) \end{matrix} \right) = 0 \; .
\>
The coefficients $N_j^i$ in (\ref{CM4}) are then defined as
\[
\label{CM4coeff}
N_j^i \coloneqq \begin{cases}
\left( \frac{\left. M_0 \right|_{2,3}}{\left. N_1^{(\mathcal{C})} \right|_{2,3}} - \frac{\widetilde{M}_0 |_{2,3}}{\left. \widetilde{N}_1^{(\mathcal{C})} \right|_{2,3}} \right)_{\lambda_0 \leftrightarrow \lambda_i^{\mathcal{B}}} \qquad \quad \mbox{for} \; j=0 \cr
\left( \frac{\left. N_j^{(\mathcal{B})} \right|_{2,3}}{\left. N_1^{(\mathcal{C})} \right|_{2,3}} - \frac{\left. \widetilde{N}_j^{(\mathcal{B})} \right|_{2,3}}{\left. \widetilde{N}_1^{(\mathcal{C})} \right|_{2,3}} \right)_{\lambda_0 \leftrightarrow \lambda_i^{\mathcal{B}}} \qquad \quad \mbox{for} \; j=1,2,3 \cr 
\end{cases} \; ,
\]
and along the lines used for $n=2$, one can also verify here that $\mbox{det} \left( N_j^i \right) \neq 0$ for generic
values of the variables. In this way we can conclude that $\mathcal{S}_3 ( \lambda_1^{\mathcal{C}} , \mu_1 - \gamma, \mu_1 | \lambda_1^{\mathcal{B}}, \lambda_2^{\mathcal{B}} , \lambda_3^{\mathcal{B}})=0$.

\paragraph{General case.} We consider $n$ levels of specializations of Eqs. (\ref{typeA}) and (\ref{typeD}),
and at each $k$-level we set $\lambda_{n-k}^{\mathcal{C}} = \mu_1 - k \gamma$ for $0 \leq k \leq n-1$, keeping the specializations at the 
previous levels. At the final level $k=n-1$ we are left with the following equations,
\<
\label{sys}
\left. M_0 \right|_{*} \mathcal{S}_n (\gen{X}^{*} | \gen{Y}^{1,n}) + \left. N_{\lambda_1^{\mathcal{C}}}^{(\mathcal{C})} \right|_{*} \mathcal{S}_n (\gen{X}^{**} | \gen{Y}^{1,n}) + \sum_{\lambda \in \gen{Y}^{1,n}}  \left. N_{\lambda}^{(\mathcal{B})} \right|_{*} \mathcal{S}_n (\gen{X}^{*} | \gen{Y}_{\lambda}^{0,n}) &=& 0 \nonumber \\
\left. \widetilde{M}_0 \right|_{*} \mathcal{S}_n (\gen{X}^{*} | \gen{Y}^{1,n}) + \left. \widetilde{N}_{\lambda_1^{\mathcal{C}}}^{(\mathcal{C})} \right|_{*} \mathcal{S}_n (\gen{X}^{**} | \gen{Y}^{1,n}) + \sum_{\lambda \in \gen{Y}^{1,n}}  \left. \widetilde{N}_{\lambda}^{(\mathcal{B})} \right|_{*} \mathcal{S}_n (\gen{X}^{*} | \gen{Y}_{\lambda}^{0,n}) &=& 0 \; , \nonumber \\
\>
where $\left. \right|_{*}$ denotes the aforementioned specializations and
\<
\gen{X}^{*} &\coloneqq& \{ \mu_1 - k \gamma \; | \; 0 \leq k \leq n-1 \} \nonumber \\
\gen{X}^{**} &\coloneqq& \{ \lambda_0  \} \cup \{ \mu_1 - k \gamma \; | \; 0 \leq k < n-1 \} \; . 
\>
Next we eliminate the term $\mathcal{S}_n (\gen{X}^{**} | \gen{Y}^{1,n})$ from the system of equations (\ref{sys})
and consider the $n$ additional equations obtained from the map $\lambda_0 \leftrightarrow \lambda_i^{\mathcal{B}}$ for $1 \leq i \leq n$.
The system of equations obtained through this procedure can then be written as
\<
\label{sysma}
\left( \begin{matrix}
N_0^0 & \dots & N_n^0 \cr
\vdots & \ddots & \vdots \cr
N_n^n & \dots & N_0^n \end{matrix} \right)
\left( \begin{matrix}
\mathcal{S}_n (\gen{X}^{*} | \gen{Y}^{1,n}) \cr
\vdots \cr
\mathcal{S}_n (\gen{X}^{*} | \gen{Y}^{0,n-1}) \end{matrix}
\right) = 0 
\>
with coefficients
\[
\label{CMGcoeff}
N_j^i \coloneqq \begin{cases}
\left( \frac{\left. M_0 \right|_{*}}{\left. N_1^{(\mathcal{C})} \right|_{*}} - \frac{\widetilde{M}_0 |_{*}}{\left. \widetilde{N}_1^{(\mathcal{C})} \right|_{*}} \right)_{\lambda_0 \leftrightarrow \lambda_i^{\mathcal{B}}} \qquad \quad \mbox{for} \; j=0 \cr
\left( \frac{\left. N_j^{(\mathcal{B})} \right|_{*}}{\left. N_1^{(\mathcal{C})} \right|_{*}} - \frac{\left. \widetilde{N}_j^{(\mathcal{B})} \right|_{*}}{\left. \widetilde{N}_1^{(\mathcal{C})} \right|_{*}} \right)_{\lambda_0 \leftrightarrow \lambda_i^{\mathcal{B}}} \qquad \quad \mbox{for} \; 1 \leq j \leq n \cr 
\end{cases} \; .
\]
Now one can verify in (\ref{sysma}) that $\mbox{det} \left( N_j^i \right) \neq 0$ for arbitrary values
of the variables, and we can conclude that $\mathcal{S}_n (\gen{X}^{*} | \gen{Y}^{1,n})=0$. This is still not the relation
we want to prove. For that we revisit each level of specialization backwards, taking into account the result obtained
at the final level, and perform a similar analysis. By doing so we finally obtain the desired
property $\mathcal{S}_n ( \mu_1 - \gamma, \mu_1 , \lambda_3^{\mathcal{C}}, \dots , \lambda_n^{\mathcal{C}} | \lambda_1^{\mathcal{B}}, \dots , \lambda_n^{\mathcal{B}}) = 0$.

\section{Solution for $n=1$}
\label{sec:Ln1}

The relations (\ref{SXY}) and (\ref{bSXY}) consist of a separation of variables induced 
by the location of certain zeroes of the scalar product $\mathcal{S}_n$.
Those zeroes are given in Lemma \ref{zeroes} and, alternatively, formulae (\ref{SXY}) and (\ref{bSXY})
can also be regarded as recurrence relations. This is due to the fact that the function $V$ corresponds
to the function $\mathcal{S}_{n-1}$ under the map $L \mapsto L-1$ and $\mu_i \mapsto \mu_{i+1}$
as shown in \Secref{sec:SOL}. Here we shall consider that $n \leq L$ in such a way that the last step of the iteration procedure
described by (\ref{SXY}) and (\ref{bSXY}) will require the solution of (\ref{typeA}, \ref{typeD}) for $n=1$ and arbitrary $L$. 
In that case the solution can be obtained through simple algebraic manipulations. The sequence of steps required to obtain 
the desired solution will be described in what follows.

The equations (\ref{typeA}) and (\ref{typeD}) for $n=1$ simply read
\<
\label{eq1}
M_0 \; \mathcal{S}_1 ( \lambda_1^{\mathcal{C}} | \lambda_1^{\mathcal{B}} ) +  N_{\lambda_1^{\mathcal{C}}}^{(\mathcal{C})} \; \mathcal{S}_1 ( \lambda_0 | \lambda_1^{\mathcal{B}}) + N_{\lambda_1^{\mathcal{B}}}^{(\mathcal{B})} \; \mathcal{S}_1 ( \lambda_1^{\mathcal{C}}| \lambda_0 ) &=& 0 \nonumber \\
\widetilde{M}_0 \; \mathcal{S}_1 ( \lambda_1^{\mathcal{C}} | \lambda_1^{\mathcal{B}} ) +  \widetilde{N}_{\lambda_1^{\mathcal{C}}}^{(\mathcal{C})} \; \mathcal{S}_1 ( \lambda_0 | \lambda_1^{\mathcal{B}}) + \widetilde{N}_{\lambda_1^{\mathcal{B}}}^{(\mathcal{B})} \; \mathcal{S}_1 ( \lambda_1^{\mathcal{C}}| \lambda_0 ) &=& 0 \; ,
\>
with coefficients explicitly given by
\<
M_0 &=& c \frac{b(2 \lambda_0) b(h + \lambda_0) b(\lambda_1^{\mathcal{C}} - \lambda_1^{\mathcal{B}}) a(\lambda_1^{\mathcal{C}} + \lambda_1^{\mathcal{B}})}{b(\lambda_0 - \lambda_1^{\mathcal{C}}) b(\lambda_0 - \lambda_1^{\mathcal{B}}) a(\lambda_0 + \lambda_1^{\mathcal{C}}) a(\lambda_0 + \lambda_1^{\mathcal{B}})} \prod_{j=1}^L a(\lambda_0 - \mu_j) a(\lambda_0 + \mu_j) \nonumber \\
N_{\lambda_1^{\mathcal{C}}}^{(\mathcal{C})} &=& c \frac{b(2 \lambda_1^{\mathcal{C}})}{a(2 \lambda_1^{\mathcal{C}})} \frac{b(h + \lambda_1^{\mathcal{C}})}{b(\lambda_1^{\mathcal{C}} - \lambda_0)} \prod_{j=1}^L a(\lambda_1^{\mathcal{C}} - \mu_j) a(\lambda_1^{\mathcal{C}} + \mu_j) \nonumber \\
&& - \; c \frac{b(2 \lambda_1^{\mathcal{C}})}{a(2 \lambda_1^{\mathcal{C}})} \frac{a(\lambda_1^{\mathcal{C}} - h)}{a(\lambda_1^{\mathcal{C}} + \lambda_0)} \prod_{j=1}^L b(\lambda_1^{\mathcal{C}} - \mu_j) b(\lambda_1^{\mathcal{C}} + \mu_j)  \nonumber \\
N_{\lambda_1^{\mathcal{B}}}^{(\mathcal{B})} &=& - \; c \frac{b(2 \lambda_1^{\mathcal{B}})}{a(2 \lambda_1^{\mathcal{B}})} \frac{b(h + \lambda_1^{\mathcal{B}})}{b(\lambda_1^{\mathcal{B}} - \lambda_0)} \prod_{j=1}^L a(\lambda_1^{\mathcal{B}} - \mu_j) a(\lambda_1^{\mathcal{B}} + \mu_j) \nonumber \\
&& + \; c \frac{b(2 \lambda_1^{\mathcal{B}})}{a(2 \lambda_1^{\mathcal{B}})} \frac{a(\lambda_1^{\mathcal{B}} - h)}{a(\lambda_1^{\mathcal{B}} + \lambda_0)} \prod_{j=1}^L b(\lambda_1^{\mathcal{B}} - \mu_j) b(\lambda_1^{\mathcal{B}} + \mu_j) 
\>
and
\<
\widetilde{M}_0 &=& c \frac{b(2 \lambda_0)}{a(2 \lambda_0)} \frac{a(2 \lambda_0 + \gamma) a(\lambda_0 - h) b(\lambda_1^{\mathcal{C}} - \lambda_1^{\mathcal{B}}) a(\lambda_1^{\mathcal{C}} + \lambda_1^{\mathcal{B}})}{b(\lambda_0 - \lambda_1^{\mathcal{C}}) b(\lambda_0 - \lambda_1^{\mathcal{B}}) a(\lambda_0 + \lambda_1^{\mathcal{C}}) a(\lambda_0 + \lambda_1^{\mathcal{B}})} \prod_{j=1}^L b(\lambda_0 - \mu_j) b(\lambda_0 + \mu_j) \nonumber \\
\widetilde{N}_{\lambda_1^{\mathcal{C}}}^{(\mathcal{C})} &=& - c \frac{b(2 \lambda_1^{\mathcal{C}})}{a(2 \lambda_1^{\mathcal{C}})} \frac{b(h + \lambda_1^{\mathcal{C}}) a(2\lambda_0 + \gamma)}{a(\lambda_1^{\mathcal{C}} + \lambda_0) b(2\lambda_0 + \gamma)} \prod_{j=1}^L a(\lambda_1^{\mathcal{C}} - \mu_j) a(\lambda_1^{\mathcal{C}} + \mu_j) \nonumber \\
&& - \; c \frac{b(2 \lambda_1^{\mathcal{C}})}{a(2 \lambda_1^{\mathcal{C}})} \frac{a(\lambda_1^{\mathcal{C}} - h) a(2\lambda_0 + \gamma)}{b(\lambda_0 - \lambda_1^{\mathcal{C}}) b(2 \lambda_0 + \gamma)  } \prod_{j=1}^L b(\lambda_1^{\mathcal{C}} - \mu_j) b(\lambda_1^{\mathcal{C}} + \mu_j) \nonumber \\
\widetilde{N}_{\lambda_1^{\mathcal{B}}}^{(\mathcal{B})} &=&  c \frac{b(2 \lambda_1^{\mathcal{B}})}{a(2 \lambda_1^{\mathcal{B}})} \frac{b(h + \lambda_1^{\mathcal{B}}) a(2\lambda_0 + \gamma)}{a(\lambda_1^{\mathcal{B}} + \lambda_0) b(2\lambda_0 + \gamma)} \prod_{j=1}^L a(\lambda_1^{\mathcal{B}} - \mu_j) a(\lambda_1^{\mathcal{B}} + \mu_j) \nonumber \\
&& + \; c \frac{b(2 \lambda_1^{\mathcal{B}})}{a(2 \lambda_1^{\mathcal{B}})} \frac{a(\lambda_1^{\mathcal{B}} - h) a(2\lambda_0 + \gamma)}{b(\lambda_0 - \lambda_1^{\mathcal{B}}) b(2 \lambda_0 + \gamma)  } \prod_{j=1}^L b(\lambda_1^{\mathcal{B}} - \mu_j) b(\lambda_1^{\mathcal{B}} + \mu_j) \; . 
\>
We then proceed by eliminating the term $\mathcal{S}_1 ( \lambda_0 | \lambda_1^{\mathcal{B}})$ from (\ref{eq1}).
By doing so we find the separated relation,
\<
\label{first}
\frac{a(2 \lambda_1^{\mathcal{B}})}{b(2 \lambda_1^{\mathcal{B}})} \frac{b(\lambda_1^{\mathcal{B}} - \lambda_1^{\mathcal{C}}) a(\lambda_1^{\mathcal{B}} + \lambda_1^{\mathcal{C}})}{\psi(\lambda_1^{\mathcal{B}} , \lambda_1^{\mathcal{C}})} \mathcal{S}_1 ( \lambda_1^{\mathcal{C}} | \lambda_1^{\mathcal{B}} ) = \frac{a(2 \lambda_0)}{b(2 \lambda_0)} \frac{b(\lambda_0 - \lambda_1^{\mathcal{C}}) a(\lambda_0 + \lambda_1^{\mathcal{C}})}{\psi(\lambda_0 , \lambda_1^{\mathcal{C}})} \mathcal{S}_1 ( \lambda_1^{\mathcal{C}} | \lambda_0 ) \; , \nonumber \\
\>
with function $\psi(\lambda , \bar{\lambda})$ reading
\<
\label{psi}
\psi(\lambda , \bar{\lambda}) &\coloneqq& b(h + \lambda) \prod_{j=1}^L a(\lambda - \mu_j) a(\lambda + \mu_j) \left[b(h + \bar{\lambda}) b(\lambda - \bar{\lambda})  \prod_{j=1}^L a(\bar{\lambda} - \mu_j) a(\bar{\lambda} + \mu_j) \right. \nonumber \\
&& \qquad \qquad \qquad \qquad \qquad \qquad \qquad  + \left. a(\bar{\lambda} -h) a(\lambda + \bar{\lambda})  \prod_{j=1}^L b(\bar{\lambda} - \mu_j) b(\bar{\lambda} + \mu_j) \right] \nonumber \\
&& - a(\lambda - h) \prod_{j=1}^L b(\lambda - \mu_j) b(\lambda + \mu_j) \left[b(h + \bar{\lambda}) a(\lambda + \bar{\lambda})  \prod_{j=1}^L a(\bar{\lambda} - \mu_j) a(\bar{\lambda} + \mu_j) \right. \nonumber \\
&& \qquad \qquad \qquad \qquad \qquad \qquad \qquad  + \left. a(\bar{\lambda} -h) b(\lambda - \bar{\lambda})  \prod_{j=1}^L b(\bar{\lambda} - \mu_j) b(\bar{\lambda} + \mu_j) \right] \; . \nonumber \\
\>
The variable $\lambda_1^{\mathcal{C}}$ plays the role of a parameter in (\ref{first}) and we can readily conclude that
\[
\label{second}
\mathcal{S}_1 ( \lambda_1^{\mathcal{C}} | \lambda_1^{\mathcal{B}} ) = \frac{b(2 \lambda_1^{\mathcal{B}})}{a(2 \lambda_1^{\mathcal{B}})} \frac{F(\lambda_1^{\mathcal{C}}) \; \psi(\lambda_1^{\mathcal{B}} , \lambda_1^{\mathcal{C}})} {b(\lambda_1^{\mathcal{B}} - \lambda_1^{\mathcal{C}}) a(\lambda_1^{\mathcal{B}} + \lambda_1^{\mathcal{C}})} \; ,
\]
where $F\colon \mathbb{C} \to \mathbb{C}$ is a function yet to be determined.

Next we solve (\ref{eq1}) for $\mathcal{S}_1 ( \lambda_1^{\mathcal{C}} | \lambda_1^{\mathcal{B}})$ instead and use formula (\ref{second}). This leaves
us with the following equation for the function $F$,
\[
\label{FF}
\frac{a(2 \lambda_1^{\mathcal{C}})}{b(2 \lambda_1^{\mathcal{C}})} F(\lambda_1^{\mathcal{C}}) = \frac{a(2 \lambda_0)}{b(2 \lambda_0)} F(\lambda_0) \; .
\]
Eq. (\ref{FF}) is readily solved by $F(\lambda) = \kappa \frac{b(2 \lambda)}{a(2 \lambda)}$ where $\kappa$ is a constant. Gathering our results
we then find the solution
\[
\label{third}
\mathcal{S}_1 ( \lambda_1^{\mathcal{C}} | \lambda_1^{\mathcal{B}} ) = c \frac{b(2 \lambda_1^{\mathcal{C}})}{a(2 \lambda_1^{\mathcal{C}})} \frac{b(2 \lambda_1^{\mathcal{B}})}{a(2 \lambda_1^{\mathcal{B}})}  \frac{\psi(\lambda_1^{\mathcal{B}} , \lambda_1^{\mathcal{C}})} {b(\lambda_1^{\mathcal{B}} - \lambda_1^{\mathcal{C}}) a(\lambda_1^{\mathcal{B}} + \lambda_1^{\mathcal{C}})} \; .
\]
In (\ref{third}) we have already considered $\kappa = c$ in accordance with the asymptotic behavior (\ref{asymp}).

\bibliographystyle{hunsrt}
\bibliography{references}

\end{document}